\documentclass[prl,twocolumn,nofootinbib,floatfix,superscriptaddress]{revtex4-1}

\usepackage{amssymb,amsfonts,amsmath,amsthm}
\usepackage{graphicx}

\usepackage{color}

\usepackage[colorlinks=true]{hyperref}

\newcommand{\bra}[1]{\langle {#1} \vert}

\newcommand{\ket}[1]{\vert {#1} \rangle}
\newcommand{\kket}[1]{\vert {#1} \rangle\!\rangle}

\newcommand{\proj}[1]{\vert {#1} \rangle\!\langle {#1} \vert}
\newcommand{\ketbra}[2]{\vert {#1} \rangle\!\langle {#2} \vert}

\newcommand{\pproj}[1]{\vert {#1} \rangle\!\rangle\!\langle\!\langle {#1} \vert}
\newcommand{\Tr}[0]{{\mathrm{Tr}}}
\newcommand{\hcal}[0]{{\mathcal{H}}}
\newcommand{\ical}[0]{{\mathcal{I}}}
\newcommand{\ocal}[0]{{\mathcal{O}}}

\newcommand{\Jamiolkowski}[0]{{Jamio\l kowski}}

\makeatletter
\newcommand{\doublewidetilde}[1]{{%
  \mathpalette\double@widetilde{#1}%
}}
\newcommand{\double@widetilde}[2]{%
  \sbox\z@{$\m@th#1\widetilde{#2}$}%
  \ht\z@=.9\ht\z@
  \widetilde{\box\z@}%
}
\makeatother
\newcommand{\ttsmap}[1]{\doublewidetilde{\mbox{$\mathcal{#1}$}}}
\newcommand{\tmap}[1]{\widetilde{\mbox{$\mathcal{#1}$}}}
\newcommand{\tmapf}[1]{\widetilde{\mbox{${#1}$}}}
\newcommand{\tmapid}[0]{\widetilde{\mbox{$id$}}}

\newcommand{\stgs}[0]{{t}}

\newcommand{\nio}{N^{\ical_0 \ical \ocal}}
\newcommand{\psym}[0]{{\Pi_{sym}^{\ical \ocal}}}
\newcommand{\pnsym}[0]{{\Pi_{sym}^{\perp}}}
\newcommand{\psup}[0]{{\Pi_{sup}}}

\newtheorem{thm}{Theorem}
\newtheorem{lem}{Lemma}

\theoremstyle{definition}
\newtheorem{rem}{Remark}

\begin{document}

\title{Success-or-Draw:~A~Strategy~Allowing~Repeat-Until-Success~in~Quantum~Computation}

\author{Qingxiuxiong Dong}
 \email{qingxiuxiong.dong@gmail.com}
 \affiliation{Department of Physics, Graduate School of Science, The University of Tokyo, Bunkyo-ku, Tokyo 113-0033, Japan}
\author{Marco T{\'u}lio Quintino}
 \email{marco.quintino@univie.ac.at}
 \affiliation{Department of Physics, Graduate School of Science, The University of Tokyo, Bunkyo-ku, Tokyo 113-0033, Japan}
 \affiliation{Vienna Center for Quantum Science and Technology (VCQ), Faculty of Physics, University of Vienna, Boltzmanngasse 5, A-1090 Vienna, Austria}
 \affiliation{Institute for Quantum Optics and Quantum Information (IQOQI), Austrian Academy of Sciences, Boltzmanngasse 3, A-1090 Vienna, Austria}
\author{Akihito Soeda}
 \email{soeda@phys.s.u-tokyo.ac.jp}
 \affiliation{Department of Physics, Graduate School of Science, The University of Tokyo, Bunkyo-ku, Tokyo 113-0033, Japan}
\author{Mio Murao}
 \email{murao@phys.s.u-tokyo.ac.jp}
 \affiliation{Department of Physics, Graduate School of Science, The University of Tokyo, Bunkyo-ku, Tokyo 113-0033, Japan}
 \affiliation{Trans-scale Quantum Science Institute, The University of Tokyo, Bunkyo-ku, Tokyo 113-0033, Japan}

\begin{abstract}

Repeat-until-success strategy is a standard method to obtain success with a probability which grows exponentially in the number of iterations.
However, since quantum systems are disturbed after a quantum measurement, it is not straightforward how to perform repeat-until-success strategies in certain quantum algorithms.
In this paper, we propose a new structure for probabilistic higher-order transformation named success-or-draw, which allows a repeat-until-success implementation.
For that we provide a universal construction of success-or-draw structure which works for any probabilistic higher-order transformation on unitary operations. 
We then present a semidefinite programming approach to obtain optimal success-or-draw protocols and analyze in detail the problem of inverting a general unitary operation.

\end{abstract}

\date{\today}

\maketitle

{\it Introduction --}
Quantum algorithms are an inevitable element for exploiting the potential of quantum computation~\cite{nielsenchuang,wilde,watrous}.
In many quantum algorithms, a unitary operation characterizing the problem and its related quantum operations, such as its inverse operation, are used as subroutines.
Quantum supermaps describe such relationships between quantum operations,
and are used for analyzing higher-order transformations between quantum operations~\cite{comb1,comb2,comb3}.
In spite of their concrete formalism, many ``useful'' supermaps, 
such as cloning unitary operations~\cite{unitary_cloning}, inverting unitary operations~\cite{unitary_inversion1,unitary_inversion2,refocusing,storage_retrieval,one_shot_entropy}, 
controlling unitary operations~\cite{controllization}, unitary learning~\cite{unitary_learning,storage_retrieval}, 
are not physically implementable in an exact and deterministic manner.
In order to perform such supermaps, two types of relaxation are usually considered: 
the approximate transformation and the probabilistic transformation.
In addition to these relaxations, adding certain resources is also considered, especially by allowing multiple calls of an input quantum operation.
In the quantum circuit implementation, multiple calls are achievable by using the corresponding quantum circuit multiple times.
With the assumption of multiple calls,
the strategy to approximate supermaps has an advantage because it is always possible to perform process tomography~\cite{process_tomography} to obtain a classical description of the input quantum operation,
calculate the output quantum operation of the supermap,
and implement the output quantum operation according to the classical description.
On the other hand, it is not known in general whether we can perform a supermap probabilistically but exactly, 
even if an arbitrary but finite number of calls are allowed.

\begin{figure}
\raisebox{1em}{a)} \includegraphics[width=.8\hsize]{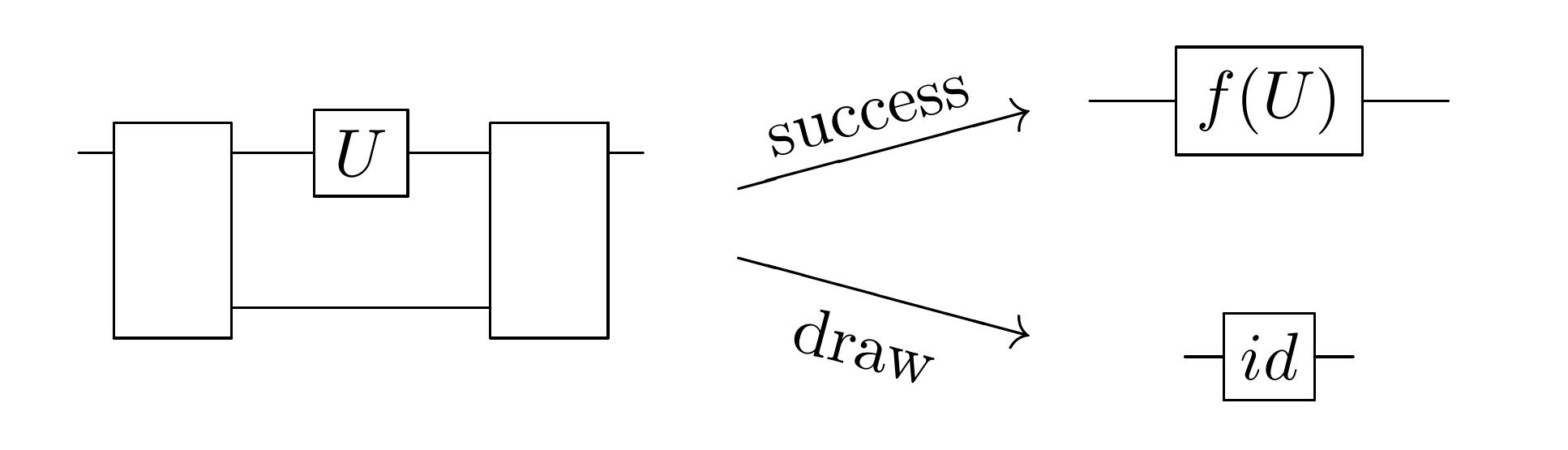} \\
\raisebox{1em}{b)} \includegraphics[width=.66\hsize]{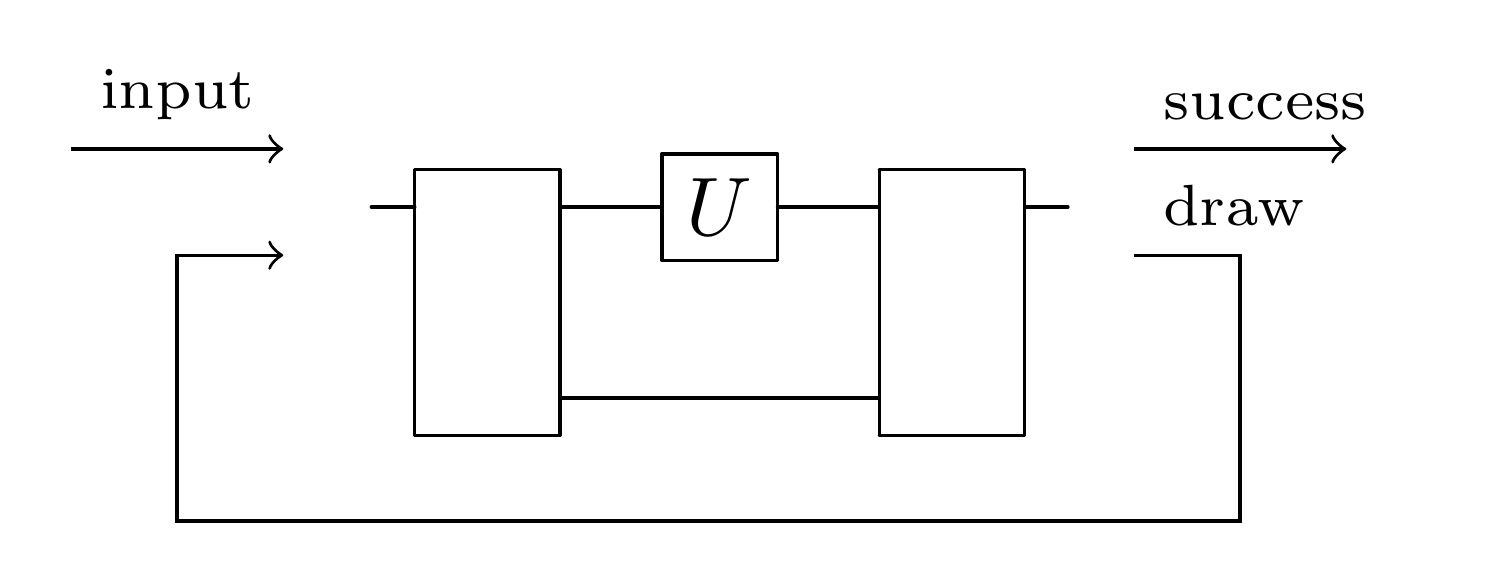}
\caption{Success-or-draw supermaps are the ones when the supermap fails, the output state remains to be the initial state.
Since the initial state is not altered on failure, one can re-iterate this protocol to obtain an exponentially decreasing failure probability.
Fig.~\ref{fig:sn_concept}a represents the action of a success-or-draw supermap: when it succeeds, the target operation $f(U)$ is obtained;
and when it draws, the identity operation $id$ is obtained.
Fig.~\ref{fig:sn_concept}b illustrates a repeat-until-success protocol which is allowed for a success-or-draw protocol.}
\label{fig:sn_concept}
\end{figure}

Quantum process tomography~\cite{process_tomography} allows a universal and approximate implementation of quantum supermaps, 
but the figure of merit, usually the average fidelity $F$, is expected to scale as $1- 1/poly(N)$ given $N$ calls of the input operation.
The probabilistic strategy, on the other hand, can achieve a success probability converges to one exponentially,
if it is possible to perform independent trials such as in a repeat-until-success protocol.
That is, if we can perform a probabilistic supermap with probability $p$ using ``a unit of resources'' such as one quantum operations,
we can perform this probabilistic supermap with probability $1-(1-p)^N$ using $N$ units of resources.
However, the resources required to perform a supermap are not only the input quantum operations,
the input state which the output quantum operation of the supermap is applied on, should also be counted as a resource.
In quantum mechanics, transformations usually disturb quantum states~\cite{wolf,information_disturbance}, 
and the input state of a probabilistic supermap is usually changed regardless of success or failure,
that is, the input state is lost after a trial of a supermap. 
Also, the cloning of a quantum state is forbidden by the no-cloning theorem~\cite{nocloning}.
Thus, it is not possible to simply perform independent trials. 
On the other hand, while allowing multiple copies of an input state may help in certain tasks~\cite{ncopy},
it is difficult to realize in many cases.
For such reason, we consider probabilistic supermaps under the following assumption,
which is also a well-studied scenario in many previous researches~\cite{unitary_cloning,unitary_inversion1,unitary_inversion2,refocusing,storage_retrieval,controllization,unitary_learning,
storage_retrieval,pbt1,pbt2,unitary_conjugate}:
multiple calls of an input quantum operation are allowed, and only a single use of an input state is available.

In this paper, we propose a structure of probabilistic supermap called ``success-or-draw'' structure as Fig.~\ref{fig:sn_concept}a shows.
In a usual probabilistic supermap, the input state is lost when it fails 
because an unknown quantum operation, which is the output quantum operation of the probabilistic supermap on failure, is applied on the input state, such as in the universal programmable quantum processor by port-based teleportation~\cite{pbt1,pbt2} and the probabilistic store and retrieve of unitary operations~\cite{storage_retrieval}.
This fact together with the impossibility to clone the quantum state makes another trial to be not possible.
However, while it is not possible to clone the quantum state, it is not known if it is possible to ``keep'' the quantum state when a probabilistic supermap fails.
Thus, we propose a probabilistic supermap which ``keeps'' the quantum state on failure, 
or we call it a draw as we are able to perform another trial when it happens as Fig.~\ref{fig:sn_concept}b shows.
To summarize, a success-or-draw supermap has the following structure:
when it is success, the target quantum operation is obtained;
when it is draw, the identity operation is obtained;
and the probability of success and draw sum up to one.

Is it always possible to find a success-or-draw supermap for a given task?
In Refs.~\cite{unitary_inversion1,unitary_inversion2}, the probabilistic unitary inversion, 
which is a supermap transforming a unitary operation into its inverse, has been analyzed in the multiple calls scenario,
and it is shown that the success-or-draw structure can be achieved by a construction of the quantum circuit.
In this paper, we show that the success-or-draw structure can be achieved for a larger class of supermaps by using a certain number of copies of an input quantum operation.
Precisely, if there exists a probabilistic supermap transforming a single $d$-dimensional unitary operation into an arbitrary completely positive and trace-preserving (CPTP) map,
then it is possible to construct a success-or-draw supermap with $d$ copies of the input unitary operation.
In particular, if a supermap is completely CP preserving (CCPP), a condition corresponding to complete positivity (CP) of quantum operations,
it is probabilistically implementable~\cite{comb3}.
Even if a supermap is not probabilistically implementable, 
in case that the supermap is linear, its approximated version, e.g., the structural physical approximation~\cite{spa1,spa2}, is probabilistically implementable.
Thus, this result applies to all linear supermaps on unitary operations if an approximation is also considered.

This result indicates that if there exists a probabilistic supermap transforming a unitary operation into a CPTP map,
then the success probability of this supermap can approach one exponentially by allowing multiple calls of the input unitary operation.
Moreover, the corresponding physical realization is given by repetitive trials of a single block of probabilistic supermap as shown in Fig.~\ref{fig:sn_concept}b,
and the cost for building the corresponding quantum circuit does not increase with the number of calls.

{\it Success-or-Draw Supermap --}
We first review the basics of supermaps.
A supermap that is using the input quantum operations in a fixed order is known as a quantum comb, 
and is the one that can be implemented in the usual quantum circuit model.
In Refs.~\cite{comb1,comb2}, a formulation of a quantum comb is presented.
In order to avoid confusion, we denote quantum operations with a tilde and supermaps with a double tilde. 
For example, given a unitary operator $U$, we denote the corresponding unitary operation by $\tmap{U}$.
A deterministic comb is described by a CCPP supermap with a set of linear constraints.
A probabilistic comb, consists of a success part and a failure part, can be described with two supermaps,
say $\ttsmap{S}$ and $\ttsmap{F}$ respectively, which sum up to a deterministic comb.

Consider the probabilistic supermap transforming unitary operations $\{ \tmap{U} \}$ into CPTP maps $\{ f(\tmap{U}) \}$.
In the usual setting of a probabilistic supermap, this problem is formulated by the constraints
\begin{gather}
\ttsmap{S} (\tmap{U}) = p_U f(\tmap{U}) \\
\ttsmap{S}, \ttsmap{F} \text{ is CCPP} \\
\ttsmap{S}+ \ttsmap{F} \text{ is a deterministic comb}.
\end{gather}

For the success-or-draw supermap, the action on failure is also determined,
and extra constraints on $\ttsmap{F}$ are required.
For the convenience for the following discussions, we also assume that we have $K$ calls to the input unitary operation $\tmap{U}$.
Since any unitary operation is transformed into the identity operation on failure,
the corresponding constraints are given by
\begin{gather}
\ttsmap{S} ( \tmap{U}^{\otimes K} ) = p_U f( \tmap{U} ) \label{eq:reqs1}\\ 
\ttsmap{N} ( \tmap{U}^{\otimes K} ) \propto \tmapid \label{eq:neutralization} \\
\ttsmap{S}, \ttsmap{N} \text{ is CCPP} \label{eq:reqs4} \\
\ttsmap{S}+ \ttsmap{N} \text{ is a deterministic comb}, \label{eq:reqs5}
\end{gather}
where $\tmapid$ denotes the identity operation, indicating that the input state does not change on failure.
Here we use $\ttsmap{N}$ instead of $\ttsmap{F}$ to denote that it corresponds to draw instead of failure, 
and this condition is also known as the neutralization condition introduced in Ref.~\cite{controllization}.

{\it Main Result --}
Theorem~\ref{thm:main} is the main result on the realizability of success-or-draw supermap.
A pictorial interpretation of Theorem~\ref{thm:main} is given by Fig.~\ref{fig:theorem_main}.

\begin{thm}\label{thm:main}
Given a probabilistic comb transforming $d$-dimensional unitary operations $\{ \tmap{U} \}$ to CPTP maps $\{ f(\tmap{U}) \}$
as $\ttsmap{S}_\stgs : \tmap{U} \mapsto p_U f(\tmap{U})$.
Then there exist $\varepsilon > 0$ and a set of probabilistic combs $\ttsmap{S}$ and $\ttsmap{N}$ that sum up to a deterministic comb,
whose actions are given by 
\begin{align}
\ttsmap{S} &: \tmap{U}^{\otimes d} \mapsto \varepsilon p_U f(\tmap{U}) \\
\ttsmap{N} &: \tmap{U}^{\otimes d} \mapsto (1-\varepsilon p_U) \tmapid.
\end{align}
\end{thm}

\begin{figure}
\includegraphics[width=.7\hsize]{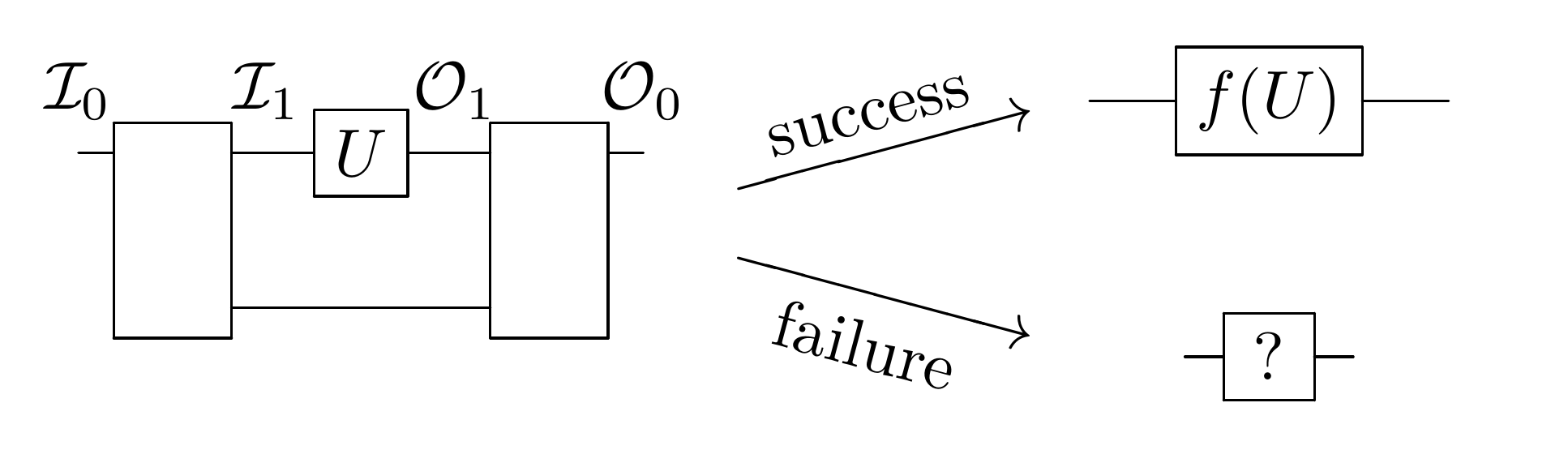} \\
$\Downarrow$ \\
\includegraphics[width=\hsize]{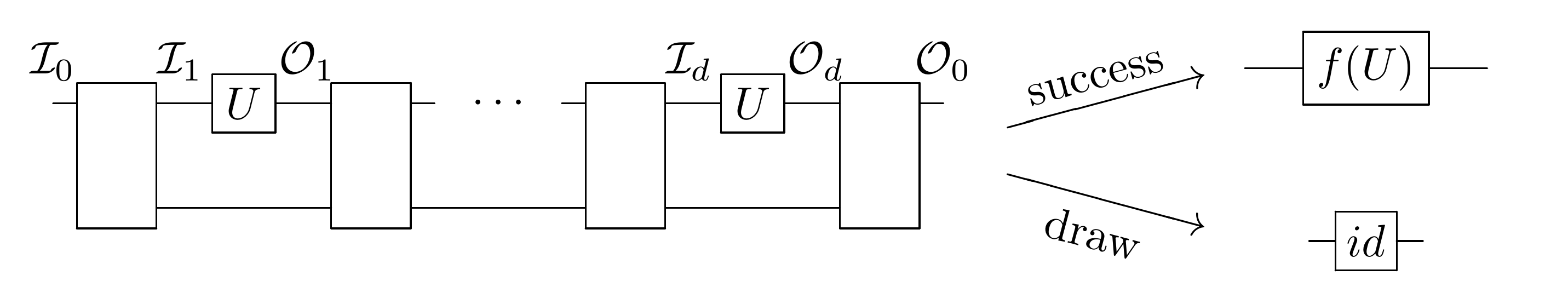}
\caption{A pictorial interpretation of Theorem~\ref{thm:main}.
We consider the case when there exists a probabilistic comb (upper) that transforms unitary operations $U$ into CPTP maps $f(U)$,
whose action is arbitrary on failure.
Theorem~\ref{thm:main} states that in this case, there exists a $d$-slot probabilistic comb (lower) that performs the same action on success,
and performs the identity operation on failure/draw, which corresponds to the preservation of the input state.}
\label{fig:theorem_main}
\end{figure}

The proof of Theorem~\ref{thm:main} is given in Appendix~\cite{sm}, which includes Ref.~\cite{ggm_basis}.
The proof is constructive.
We present a construction of the combs $\ttsmap{S}$ and $\ttsmap{N}$ from the comb $\ttsmap{S}_\stgs$,
more precisely, we show a construction of $S$ and $N$, the Choi operators~\cite{jamiolkowski,choi,comb1,comb2} of $\ttsmap{S}$ and $\ttsmap{N}$, from $S_\stgs$, the Choi operator of $\ttsmap{S}_\stgs$.
The requirements for the combs are given by Eqs.~\eqref{eq:reqs1}-\eqref{eq:reqs5},
which need to be satisfied simultaneously.

While we only require that $S+N$ is a deterministic comb, that is, the input operations are used in a sequential way,
the construction shown in the proof (Eq.~\eqref{eq:construction_nio} of Appendix~\cite{sm}) satisfies an extra condition
\begin{align}
\Tr_{\ocal_0} (S+N) = \Tr_{\ocal_2 \ocal_3 \cdots \ocal_d \ocal_0} (S+N) \otimes \frac{ I_{\ocal_2 \ocal_3 \cdots \ocal_d} }{ d^{d-1} },
\end{align}
where $\ocal_0, \ldots, \ocal_d$ are the Hilbert spaces as shown in Fig.~\ref{fig:theorem_main}.
This condition shows that the comb can be decomposed into two blocks as the quantum circuit shown in Fig.~\ref{fig:depth_two}:
the first block uses only a single unitary operation, while the second block uses the remaining $d-1$ unitary operations in parallel.
Such a structure indicates that while the number of calls increases with $d$, the depth of this comb is constant as two.
Note that we can assume this structure if we only consider non-zero success probability,
and in general, adding this assumption would decrease the success probability.

\begin{figure}
\includegraphics[width=0.6\hsize]{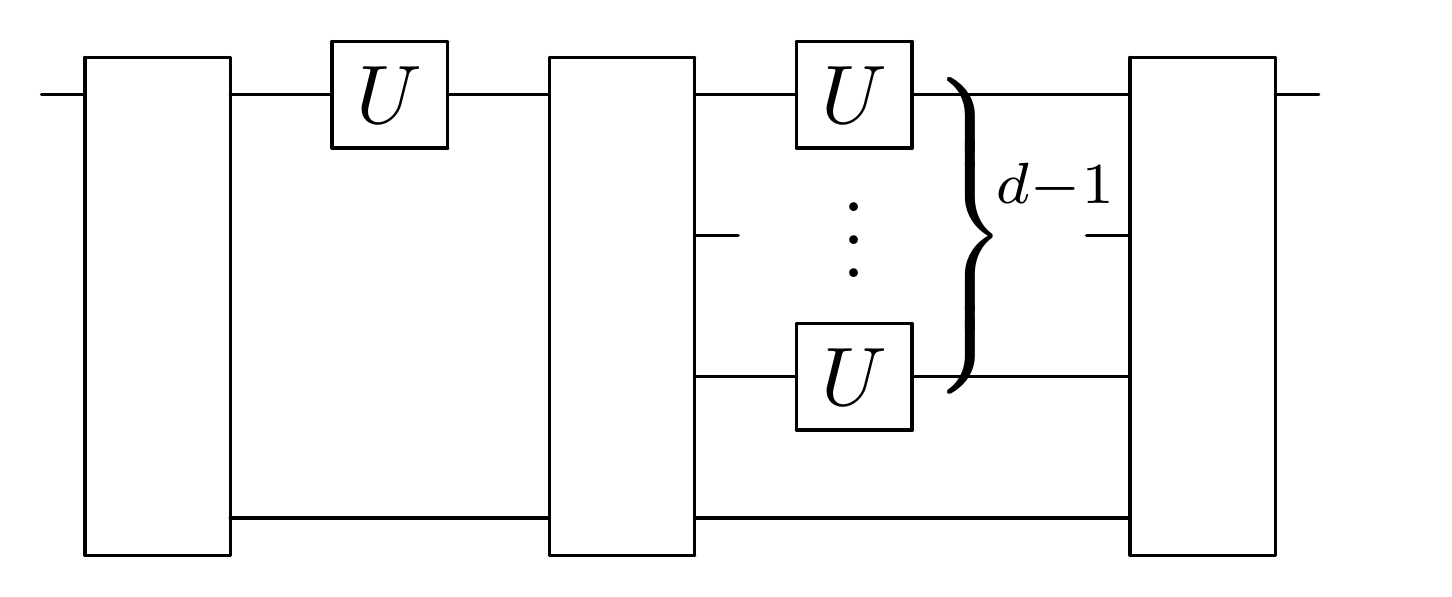}
\caption{The constructed success-or-draw comb in Appendix~\cite{sm} has an extra structure:
it uses one copy of the unitary operation at first, and then uses the remaining $d-1$ copy of the unitary operation in parallel.}
\label{fig:depth_two}
\end{figure}

When the indefinite causal order~\cite{indefinite1,indefinite2,indefinite3,indefinite_purification} is considered, 
the construction can be replaced by a simpler one by exploiting the symmetry as Remark~\ref{rem:indefinite_causal}, 
and a higher success probability can be achieved in general.

{\it Unitary Inversion --}
In this section, we analyze the probabilistic unitary inversion as a success-or-draw supermap.
We only consider the two-dimensional case $d=2$ here.
The optimal success probability can be obtained by the following SDP
\begin{gather}
\max\quad p \\
\mathrm{s.t.}\quad \Tr_{\ical \ocal} [ S ( {J_{U_i}}^{\otimes K} )^T ] = p {J_{U_i^{-1}}} \\
\Tr_{\ical \ocal} [ N ( {J_{U_i}}^{\otimes K} )^T ] \leq d^K {J_{id}} \\
S \geq 0, N \geq 0\\
S+N \text{ is a deterministic comb},
\end{gather}
where $S$ and $N$ denote the Choi operators of the combs corresponding to success and draw,
$J_{\Lambda}$ denotes the Choi operator of a quantum operation $\Lambda$,
and $\{ U_i \}$ is a finite set of unitary operators that the corresponding Choi operators form 
a basis of the linear span of $\mathrm{span} \{ J_U^{\otimes k} \}$ (see Refs.~\cite{unitary_inversion1, unitary_inversion2}).

For $K=2$, Theorem~\ref{thm:main} indicates that the optimal success probability is positive as $p>0$.
In fact, a numerical solution to this SDP shows that the optimal success probability is $p =1/3$.
This problem is also considered in Ref.~\cite{unitary_inversion1}, 
where an explicit quantum circuit with the success-or-draw structure was presented, which success probability is $1/4$.

For comparison, we briefly state the protocol presented in Ref.~\cite{unitary_inversion1}.
The protocol is similar to the teleportation protocol, which generates the state $\sigma_x^i \sigma_z^j \ket{\psi}$ before correction, 
where $\ket{\psi}$ is the initial state and  $(i,j) = (0,0),(0,1),(1,0),(1,1)$ is the outcome of the Bell measurement.
For the two-dimensional unitary inversion protocol presented in Ref.~\cite{unitary_inversion1},
we can obtain $U^{-1} \sigma_x^i \sigma_z^j \ket{\psi}$ with a single use of $U$ by a small modification to the teleportation or  gate teleportation protocol.
This protocol successfully achieves unitary inversion when $(i,j) = (0,0)$.
When it fails, on the other hand, we can obtain the state $\sigma_x^i \sigma_z^j \ket{\psi}$ by an extra use of $U$,
and the input state $\ket{\psi}$ can be recovered by applying $(\sigma_x^i \sigma_z^j)^{-1}$, which achieves the neutralization supermap.

One difference between the optimal success-or-draw protocol we obtained and the protocol presented in Ref.~\cite{unitary_inversion1} is that the latter is not only a success-or-draw protocol,
but it has another feature: it can be regarded as a success-or-resetting protocol.
The latter protocol uses a single copy of a unitary operation to obtain its inverse,
and when it fails, it results in a state that is ``resettable'' to be the input state by another unitary operation.
Such a success-or-resetting protocol may have an advantage as we can choose whether to continue the protocol by resetting after we know if it succeeded.

For $K=1$, we also prove that the optimal success probability is $p=0$,
which means it is not possible to have a success-or-draw protocol.
This result gives an explicit example that a success-or-draw protocol is not available.
The proof is given in Appendix~\cite{sm}.

\begin{figure}
\includegraphics[width=0.9\hsize]{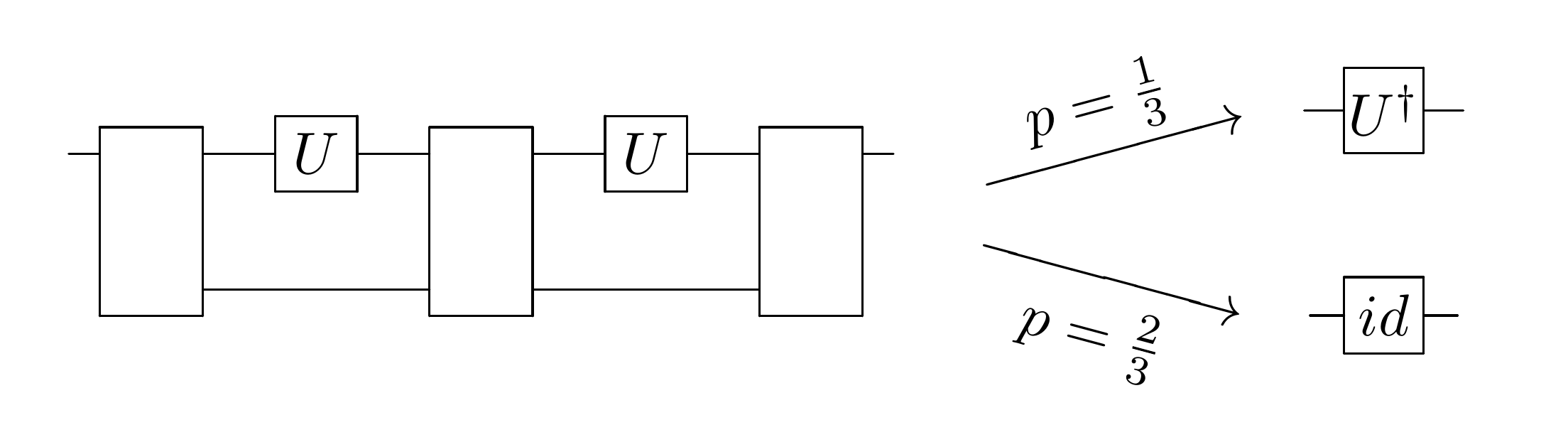}
\caption{The success-or-draw protocol for unitary inversion.
When the unitary operation can be used twice, the optimal success probability for success is $1/3$,
whereas that of the optimal success-or-resetting protocol is $1/4$.
In either case, the output on failure is the identity operation, which means the initial input state is preserved 
and it is possible to run the same protocol again with extra uses of the input unitary operation.}
\end{figure}

{\it Discussions --}
We have introduced a new structure for probabilistic supermap which we name success-or-draw structure.
A probabilistic supermap with the success-or-draw structure can amplify its success probability by more calls of the input quantum operation in a sequential manner, which scales exponentially to one in the number of calls.
A mathematical formulation for the success-or-draw supermap was presented.
We considered the case where the input quantum operation is a unitary operation, and we proved that
any probabilistic supermap transforming unitary operations into CPTP maps can become a success-or-draw supermap by adding the number of copies of the unitary operation.

We then analyzed the problem of the two-dimensional unitary inversion.
When two copies of an input unitary operation are allowed,
Theorem~\ref{thm:main} guarantees the existence of a non-trivial solution to this problem, 
and we also obtained the optimal solution numerically using SDP.
A success-or-draw protocol for this problem was also presented previously in Ref.~\cite{unitary_inversion1},
and our numerical calculation shows that a higher success probability can be achieved if we only require the success-or-draw structure.
We also proved that a success-or-draw protocol does not exist with a single copy of an input unitary operation.

Our result shows an advantage of probabilistic supermap, as it allows a success probability exponentially close to one in the sequential case.
The number of calls is also undetermined for a success-or-draw protocol, and an average number of calls may become a suitable measure in practice.
We also hope that the success-or-draw structure helps in a simpler physical implementation of supermaps.

\section*{Acknowledgments}

This work was supported by MEXT Quantum Leap Flagship Program (MEXT Q-LEAP) Grant Number JPMXS0118069605 and JPMXS0120351339, 
Japan Society for the Promotion of Science (JSPS) by KAKENHI grant No.~17H01694, 18H04286, 18K13467, 
and Advanced Leading Graduate Course for Photon Science (ALPS).
M.~T.~Quintino also acknowledges the Austrian Science Fund (FWF) through the SFB project ``BeyondC'' (sub-project F7103), 
a grant from the Foundational Questions Institute (FQXi) as part of the Quantum Information Structure of Spacetime (QISS) Project (qiss.fr).
The opinions expressed in this publication are those of the authors and do not necessarily reflect the views of the John Templeton Foundation.
This project has received funding from the European Union's Horizon 2020 research and innovation programme under the Marie {Sk\l odowska}-Curie grant agreement No.~801110 and the Austrian Federal Ministry of Education, Science and Research (BMBWF). 
It reflects only the authors' view, the EU Agency is not responsible for any use that may be made of the information it contains.

\bibliographystyle{utphys}
\bibliography{refs}

\clearpage

\onecolumngrid
\appendix

\renewcommand{\theequation}{S\arabic{equation}}
\setcounter{equation}{0}

\section{Appendix}

\section{Proof of Theorem~\ref{thm:main}}\label{ap:proof1}

\begin{proof}[Sketch of the proof]
In order to prove Theorem~\ref{thm:main}, we first prove Lemma~\ref{lem:neutralization} and Lemma~\ref{lem:unitary2cptp}, 
which indicate that it is enough to prove Theorem~\ref{thm:main2}.
Here we state the sketch of the proof.

Lemma~\ref{lem:neutralization} gives a sufficient condition of the neutralization condition Eq.~\eqref{eq:neutralization}.
The neutralization condition Eq.~\eqref{eq:neutralization} is difficult to use for many reasons,
for example, the probability for neutralization is not constant in general. 
In Theorem~\ref{thm:main}, the probability of neutralization can depend on $U$.
A direct way to rewrite Eq.~\eqref{eq:neutralization} is to add new variables $\{ q_U \}$ that correspond to the probability depend on $U$ and rewrite as
\begin{gather}
\ttsmap{N} ( \tmap{U}^{\otimes K} ) = q_U \cdot \tmapid.
\end{gather}
Since the corresponding Choi operators are positive, and that for r.h.s. is a rank-1 operator, 
this condition can be reduced to an inequality of the form
\begin{gather}
\ttsmap{N} ( \tmap{U}^{\otimes K} ) \leq c \cdot \tmapid,
\end{gather}
where $c$ is a constant determined by the normalization conditions.
This condition is equivalent to the one given by Eq.~\eqref{eq:neutralization}, 
but it is still difficult to analyze because it is necessary to consider all unitary operations.
Note that in numerical analysis, it is possible to use this condition directly, as we will state in the analysis for the unitary inversion.
In Lemma~\ref{lem:neutralization}, we show a sufficient condition by considering a symmetric subspace,
that is, $\tmap{U}^{\otimes K}$ is invariant under permutations of each input operations.

Lemma~\ref{lem:unitary2cptp} gives a characterization of the Choi operator of a probabilistic comb transforming unitary operations to CPTP maps,
which is the assumption of Theorem~\ref{thm:main}.
We consider a Hermitian basis which consists of an identity operator and traceless operators, 
and show that the decomposition of the corresponding Choi operator consists of only certain terms.
Using a basis with an identity operator and traceless operators is convenient for considering the causal condition,
because the causal condition is usually given by a set of equations consist of partial traces,
and the traceless terms help in determining which terms do not affect the causal condition.

By considering Lemma~\ref{lem:neutralization} and Lemma~\ref{lem:unitary2cptp}, 
it is enough to prove Theorem~\ref{thm:main2} in order to prove Theorem~\ref{thm:main}.
The proof of Theorem~\ref{thm:main2} can be further divided into two parts:
the first part presents a construction of the Choi operators $S$ and the partial trace of $N$ given by $\nio := \Tr_{\ocal_0} N$ from $S_\stgs$;
the second part is mainly separated into Lemma~\ref{lem:nio2n}, which presents a construction of $N$ from $\nio$.

In the first part of the proof, we first present a trivial set of Choi operators $S$ and $F$ from $S_\stgs$, 
where $F$ is a Choi operator which does not necessarily satisfy the neutralization condition Eq.\eqref{eq:neutralization} for $N$, 
but satisfies all the remaining conditions given by Eqs.~\eqref{eq:reqs1},\eqref{eq:reqs4},\eqref{eq:reqs5}.
Moreover, $F$ also has a similar decomposition given by Lemma~\ref{lem:unitary2cptp}.
We then present a construction of $\nio$ from $F$, 
where the neutralization condition is also satisfied in addition to the positivity Eq.~\eqref{eq:reqs4} and the causal conditions Eq.~\eqref{eq:reqs5}. 
The positivity of $\nio$ is satisfied by taking the operator to be a strictly positive full-rank operator,
and the main difficulty is to satisfy the causal condition and the neutralization condition simultaneously.
The decomposition given by Lemma~\ref{lem:unitary2cptp} is convenient for the causal condition
in the sense that it is possible to add certain traceless terms that do not affect the causal condition,
and we give a class of Choi operators that satisfies the causal condition.
Then, we show that among this class of Choi operators, 
it is possible to cancel the terms that do not satisfy the neutralization condition 
by using the properties of the symmetric subspace considered in Lemma~\ref{lem:neutralization}.
Thus, it is possible to satisfy the causal condition and the neutralization condition simultaneously.

In the second part of the proof, we construct $N$ from $\nio$.
In this part, the causal condition and the neutralization condition are easily satisfied because the condition is similar to the first part.
On the other hand, the positivity condition becomes difficult.
Unlike in the first part, since the target operation is the identity channel, which Choi operator is rank-1, 
we cannot take the Choi operator $N$ to be a full-rank operator, which is robust in positivity.
To solve this problem, we consider a subspace of the Hilbert space that $N$ is on, 
and we show a construction of $N$ that lies in this subspace and is a strictly positive full-rank operator in the subspace.
Thus, the positivity of $N$ can be satisfied.

\end{proof}

We first clarify the following notations.
In the following, we describe quantum operations and quantum supermaps by the Choi operators, defined via Choi-{\Jamiolkowski} isomorphism~\cite{jamiolkowski,choi}.
For a quantum operation $\tmapf{\Lambda}: L(\hcal_{in}) \to L(\hcal_{out})$, the corresponding Choi operator is given by
\begin{align}
J_\Lambda^{\hcal_{in} \hcal_{out}} := \sum_{ij} \ketbra{i}{j} \otimes \tmapf{\Lambda}( \ketbra{i}{j}) \in L(\hcal_{in} \otimes \hcal_{out}),
\end{align}
where $\{ \ket{i} \}$ is an orthonormal basis for $\hcal_{in}$.
In this paper, the Hilbert spaces of an operator are usually denoted as superscript, and may be omitted when it is trivial from the context.
The condition that a map $\tmapf{\Lambda}$ is completely positive (CP) corresponds to the positivity of $J_\Lambda$ as $J_\Lambda \geq 0$,
and the condition that it is trace-preserving (TP)  corresponds to $\Tr_{\hcal_{out}} J_\Lambda = I^{\hcal_{in}}$.
In this paper, since unitary operations play an important role, 
we also denote a unitary operation $\tmap{U}$ with the corresponding unitary operator $U$,
and its Choi operator as $J_U := \sum_{ij} \ketbra{i}{j} \otimes \tmap{U}( \ketbra{i}{j})$.
Note that a unitary operation is also a unital map, which satisfies $\Tr_{\hcal_{in}} J_\Lambda = I^{\hcal_{out}}$ in addition to the condition for a CPTP map.
The identity channel $\tmapid$ also plays an important role,
and we denote the corresponding Choi operator as $J_{id}$ instead of $J_I = d \phi^+ = d \proj{\phi^+}$, 
where $\ket{\phi^+} = \sum_{i} (1/\sqrt{d}) \ket{ii}$ is the maximally entangled state.
The action of a quantum operation $J_\Lambda$ on a quantum state $\rho$ is given by 
$\Tr_{\hcal_{in}} [J_\Lambda (\rho^T \otimes I^{\hcal_{out}} ) ]$.
In this paper, we also omit the identity operator, such as $\Tr_{\hcal_{in}} [J_\Lambda \rho^T ]$, when it is trivial from the context for convenience.

Next, we consider a $K$-slot deterministic quantum supermap $\ttsmap{C}: [\ical \to \ocal] \to [ \ical_0 \to \ocal_0]$,
where the Hilbert spaces are represented by $\ical_0, \ical_1, \ocal_1, \ldots, \ocal_K, \ocal_0$ as shown in Fig.~\ref{fig:theorem_main} in which case $K = d$,
and $\ical$ and $\ocal$ are the abbreviations of $\ical := \ical_1 \ical_2 \cdots \ical_K$ and  $\ocal := \ocal_1 \ocal_2 \cdots \ocal_K$.
In this paper, we also assume that $d_{\ical_i} = d_{\ocal_j} =: d$ for $i,j \geq 1$, and $d_{\ical_0} = d_{\ocal_0} =: d_0$.
The corresponding Choi operator $C$ is defined via Choi-{\Jamiolkowski} isomorphism as shown in Refs.~\cite{comb1,comb2}.
The completely CP preserving (CCPP) condition of $\ttsmap{C}$ is given by the positivity $C \geq 0$ similar to the quantum operation case.
The condition that the supermap uses the input operations in a fixed order, 
also known as the causal condition that it is a sequential comb, is given by the set of equations
\begin{align}
\Tr_{\ocal_0} C &= C^{(K)} \otimes \frac{I^{\ocal_{K}}}{d}  \label{eq:causal_condition_c_1} \\
\Tr_{\ical_k} C^{(k)} &= C^{(k-1)} \otimes \frac{I^{\ocal_{k-1}}}{d}  \quad  (2 \leq k \leq K) \label{eq:causal_condition_c_2} \\
\Tr_{\ical_1} C^{(1)} &= (\Tr C ) \frac{I^{\ical_0}}{d_0}   \label{eq:causal_condition_c_3}
\end{align}
where $C^{(K)} := \Tr_{\ocal_K \ocal_0} C$ and $C^{(k-1)} := \Tr_{\ocal_{k-1} \ical_k} C^{(k)}$ for $k = 2, \ldots, K$.
The normalization condition is chosen to be $\Tr\, C = d_{\ical_0} d_{\ocal}$ for convenience.
For example, $C := I^{\ical_0} \otimes \frac{I^{\ical_1}}{d_{\ical_1}} \otimes I^{\ocal_1} \otimes \cdots \otimes I^{\ocal_K} \otimes \frac{I^{\ocal_0}}{d_{\ocal_0}}$ is a deterministic comb.
Note that in our problem, it is required that $d_{\ical_i} = d_{\ocal_j} = d$ for $i,j \geq 1$, and $d_{\ical_0} = d_{\ocal_0} = d_0$.
The action of a quantum supermap $C$ on $K$ copies of a quantum operation $J_\Lambda^{\otimes K}$ is given by 
$\Tr_{\ical \ocal} [C (J_\Lambda^{\otimes K})^T ]$, 
where we omitted the identity operator of $J_\Lambda^{\otimes K} \otimes I^{\ical_0 \ocal_0}$.

For a probabilistic supermap $\ttsmap{S}: [\ical \to \ocal] \to [ \ical_0 \to \ocal_0]$,
the condition that the corresponding Choi operator $S$ satisfies is given by the following:
there exists an operator $F \geq 0$, which corresponding to the supermap on failure $\ttsmap{F}$, that $S + F$ is a deterministic supermap,
i.e., $C = S+F$ satisfies the conditions stated above.
While we use the word success and failure here, there is no mathematical difference for $S$ and $F$ except that the action of $S$ is the target supermap given by $\Tr_{\ical \ocal} [S (J_\Lambda^{\otimes K})^T ]$ as we require when the input operations are $K$ copies of $J_\Lambda$. 
For the operator $F$ corresponding to failure, 
the action is given in a similar way by $\Tr_{\ical \ocal} [F (J_\Lambda^{\otimes K})^T ]$, on which we do not have any constraint in general.
However, when we assume this operator to be proportional to the Choi operator of the identity channel $J_{id}$, 
which we call as the neutralization condition, this probabilistic supermap becomes a success-or-draw supermap.
In this case, we denote the corresponding supermap as $\ttsmap{N}$ and the Choi operator $N$ to clarify that they correspond to a neutralization supermap.

For Lemma~\ref{lem:neutralization},
we define the following operators.
We first define the permutation operator $P^\ical_\sigma$ and $P^\ocal_\sigma$ that permute systems $\ical$ and $\ocal$ according to the permutation $\sigma$.
The permutation of input operations is given by $P^{\ical \ocal}_\sigma := P^{\ical}_\sigma \otimes P^{\ocal}_\sigma$,
which simultaneously permutes the input system and the output system of a single input operation according to the permutation $\sigma$.
The symmetric subspace of input operations $\psym$ is given by
\begin{align}
\psym := \sum_\sigma P^{\ical \ocal}_\sigma = \sum_\sigma P^{\ical}_\sigma \otimes P^{\ocal}_\sigma.
\end{align}

For Lemma~\ref{lem:unitary2cptp},
we define a set of Hermitian operators $\{ g_i \}_{i=0}^{d^2-1}$ that forms the operator basis for $d$-dimensional Hermitian operators,
with $g_0 := I_d$, others being traceless, and the orthogonality $\Tr g_i g_j = d \delta_{ij}$~\cite{ggm_basis}.
For example, the Pauli matrices for $d=2$, and Gell-Mann matrices for $d=3$.
We also define the set for $d_0$-dimensional Hermitian operators as $\{ h_i \}_{i=0}^{d_0^2-1}$.
In Lemma~\ref{lem:unitary2cptp}, we rewrite the condition that a comb transforms unitary operations to CPTP maps
in terms of Choi operator and the Hermitian operator basis.

In order to prove Theorem~\ref{thm:main}, we first consider Lemma~\ref{lem:neutralization} and Lemma~\ref{lem:unitary2cptp}, 
which shows that it is enough to prove Theorem~\ref{thm:main2}.

\begin{lem}\label{lem:neutralization}
If $\Tr_{\ical \ocal} (\psym N \psym) \propto J_{id}$,
then $N$ neutralizes all unitary operations as $\ttsmap{N} ( \tmap{U}^{\otimes K}) \propto \tmapid$.
\end{lem}

\begin{proof}
Note that the if condition is equivalent to $\Tr_{\ical \ocal} (\psym N \psym) \leq d^K J_{id}$, 
because of the normalization condition $\Tr N \leq d^K d_0$.

For any input channel ${J_U}$, ${J_U}^{\otimes K} \leq d^K I$ holds and $A := d^K I - {J_U}^{\otimes K} \geq 0$. Thus
\begin{align}
J_{id} &\geq \Tr_{\ical \ocal} [(\psym N \psym) / d^K] \\
&= \Tr_{\ical \ocal} [(\psym N \psym) (A + {J_U}^{\otimes K})^T /d^{2K}] \\
&= \Tr_{\ical \ocal} [(\psym N \psym) A^T /d^{2K}] + \Tr_{\ical \ocal} [(\psym N \psym) ({J_U}^{\otimes K})^T /d^{2K}]
\end{align}
holds.
Here both $\Tr_{\ical \ocal} [(\psym N \psym) A^T / d^{2K}]$ and $\Tr_{\ical \ocal} [(\psym N \psym) ({J_U}^{\otimes K})^T/d^{2K}]$ are positive operators,
and $J_{id}$ is a rank-1 operator,
the inequality can be reduced to
\begin{align}
\Tr_{\ical \ocal} [(\psym N \psym) ({J_U}^{\otimes K})^T/d^{2K}] \leq J_{id}.
\end{align}
Since ${J_U}^{\otimes K} $ is in the symmetric subspace, that is, ${J_U}^{\otimes K}  = \psym {J_U}^{\otimes K} \psym$,
we obtain 
\begin{align}
 \Tr_{\ical \ocal} [N  ({J_U}^{\otimes K})^T] = \Tr_{\ical \ocal} [(\psym N \psym)  ({J_U}^{\otimes K})^T] \leq d^{2K} J_{id},%
\end{align}
that is, $\ttsmap{N} ( \tmap{U}^{\otimes K}) \propto \tmapid$.
\end{proof}

\begin{lem}\label{lem:unitary2cptp}
If a one-slot probabilistic comb $S_\stgs^{\ical_0 \ical_1 \ocal_1 \ocal_0}$ transforms unitary operations to CPTP maps,
then $S_\stgs^{\ical_0 \ical_1 \ocal_1} := \Tr_{\ocal_0} S_\stgs^{\ical_0 \ical_1 \ocal_1 \ocal_0}$ has a decomposition satisfying
\begin{align}
S_\stgs^{\ical_0 \ical_1 \ocal_1}  = \frac{I^{\ical_0} }{ d_0 } \otimes \Tr_{\ical_0} S_\stgs^{\ical_0 \ical_1 \ocal_1} 
+ \sum_{i=1}^{d_0^2-1} \sum_{j=1}^{d^2-1} \alpha_{ij} h_i^{\ical_0} \otimes [ g_j^{\ical_1} \otimes I^{\ocal_1} ] 
+ \sum_{i=1}^{d_0^2-1} \sum_{j=1}^{d^2-1} \beta_{ij} h_i^{\ical_0} \otimes [ I^{\ical_1} \otimes g_j^{\ocal_1} ],
\end{align}
where $\{ \alpha_{ij} \}$ and  $\{ \beta_{ij} \}$ are real coefficients.
\end{lem}

\begin{proof}

The Choi operator $S_\stgs^{\ical_0 \ical_1 \ocal_1}$ can always be decomposed as
\begin{align}
S_\stgs^{\ical_0 \ical_1 \ocal_1}  &=  \frac{I^{\ical_0} }{ d_0 } \otimes \Tr_{\ical_0} S_\stgs^{\ical_0 \ical_1 \ocal_1} 
+ \sum_{i=1}^{d_0^2-1} \sum_{j=1}^{d^2-1} \alpha_{ij} h_i^{\ical_0} \otimes [ g_j^{\ical_1} \otimes I^{\ocal_1} ] 
+ \sum_{i=1}^{d_0^2-1} \sum_{j=1}^{d^2-1} \beta_{ij} h_i^{\ical_0} \otimes [ I^{\ical_1} \otimes g_j^{\ocal_1} ] \notag \\
&+ \sum_{i=1}^{d_0^2-1} \sum_{j,k=1}^{d^2-1} \gamma_{ijk} h_i^{\ical_0} \otimes [ g_j^{\ical_1} \otimes g_k^{\ocal_1} ], \label{eq:lemma2_general_decomposition}
\end{align}
and it is enough to show that $\gamma_{ijk} = 0$ for all $i,j,k \geq 1$.

From the assumption, 
$\Tr_{\ical_1 \ocal_1} [ S_\stgs^{\ical_0 \ical_1 \ocal_1 \ocal_0} (J_U^T)^{\ical_1 \ocal_1} ] $ is proportional to the Choi operator of a CPTP map, which satisfies
\begin{align}
\Tr_{\ocal_0} \Tr_{\ical_1 \ocal_1} [ S_\stgs^{\ical_0 \ical_1 \ocal_1 \ocal_0} (J_U^T)^{\ical_1 \ocal_1} ] \propto I^{\ical_0},
\end{align}
where $I$ is the partial trace of the Choi operator of a CPTP map. Thus, $S_\stgs^{\ical_0 \ical_1 \ocal_1}$ satisfies
\begin{align}
 \Tr_{\ical_1 \ocal_1} [S_\stgs^{\ical_0 \ical_1 \ocal_1} (J_U^T)^{\ical_1 \ocal_1} ] \propto I^{\ical_0}.
\end{align}
Moreover, for any operator $O$ in the linear span of $\mathrm{span}\{ J_U \} := \{ O \mid O = \sum_i c_i J_{U_i}, c_i \in \mathbb{C} \}$, the condition
\begin{align}
 \Tr_{\ical_1 \ocal_1} [S_\stgs^{\ical_0 \ical_1 \ocal_1} (O^T)^{\ical_1 \ocal_1} ] \propto I^{\ical_0}
\end{align}
holds because of the linearity.

Next, we show that $g_j \otimes g_k \in \mathrm{span}\{ J_U \}$ for all $j,k \geq 1$ are in the linear span of $\mathrm{span}\{ J_U \}$.
From Lemma~\ref{lem:span_dim}, the dimension of the linear span is given by $\dim \mathrm{span} \{ J_U \} = (d^2 - 1)^2 + 1 $,
and one basis for this span is given by $g_j \otimes g_k$ with $j,k \geq 0$. Note that $g_0 = I_d$.
On the other hand, $g_i \otimes I$ and $I \otimes g_i$ for $i \geq 1$ are not in $\mathrm{span} \{ J_U \} $,
because of the trace-preserving property and the unitality of unitary operations, respectively.
Thus, the remaining $d^4 - 2(d^2-1) = (d^2-1)^2 + 1$ elements, especially $g_j \otimes g_k$ with $j,k \geq 1$ and $I \otimes I$, are in the linear span of $\mathrm{span}\{ J_U \}$.

Since $g_j \otimes g_k \in \mathrm{span}\{ J_U \}$ for all $j,k \geq 1$,
by substituting $S_\stgs^{\ical_0 \ical_1 \ocal_1}$ with the decomposition Eq.~\eqref{eq:lemma2_general_decomposition},
we obtain $\sum_i \gamma_{ijk} h_i^{\ical_0} \propto I^{\ical_0}$ for all $j,k \geq 1$.
Thus, $\gamma_{ijk} = 0$ is required for all $i,j,k \geq 1$, which proves the Lemma.

\end{proof}

\begin{lem}\label{lem:span_dim}
The dimension of the linear span of $\mathrm{span} \{ J_U \} := \{ O \mid O = \sum_i c_i J_{U_i}, c_i \in \mathbb{C} \} $ is $(d^2-1)^2 + 1$.
\end{lem}

\begin{proof}

We first define $\kket{U} := (U \otimes I) \kket{I}$ with $\kket{I} := \sqrt{d} \ket{\phi^+} = \sum_{i=0}^{d} \ket{ii}$ the unnormalized maximally entangled state.
The vectorization of $J_U = \pproj{U} = (U \otimes I) \pproj{I} (U^\dag \otimes I)$ is given by 
$(U^\dag \otimes I)^T \kket{I} \otimes (U \otimes I) \kket{I} = \kket{U^*} \otimes \kket{U}$,
and the dimension of $\mathrm{span} \{ J_U \}$ is equivalent to the dimension of 
$\mathrm{span} \{ \kket{U^*} \otimes \kket{U} \} := \{ O \mid O = \sum_i c_i \kket{U^*} \otimes \kket{U}, c_i \in \mathbb{C} \}$.
In order to obtain the dimension, 
we consider the projector of $\kket{U^*} \otimes \kket{U}$, and integrate over all unitary operations $U$ as 
\begin{align}
Q &= \int dU (\pproj{U^*} \otimes \pproj{U}), \label{eq:ap_def_q_int}
\end{align}
and the dimension is given by the rank of $Q$.
Consider the substitution of $U \to VU$ with arbitrary $V$ and the invariance of the Haar measure, $Q$ satisfies
\begin{align}
Q &= \int dU (V^* \otimes I \otimes V \otimes I) (\pproj{U^*} \otimes \pproj{U}) (V^T \otimes I \otimes V^\dag \otimes I) \\
&= (V^* \otimes I \otimes V \otimes I) Q (V^T \otimes I \otimes V^\dag \otimes I).
\end{align}
For convenience, we denote the space that $V$ and $V^*$ acting on by $A$ and the remaining by $B$,
then $Q$ satisfies the commutation relation
\begin{align}
[ Q , (U^* \otimes U)^{A} \otimes I^{B} ] = 0
\end{align}
for all unitary operators $U$.
The irreducible representation of $(U^* \otimes U)$ is given by
\begin{align}
U^* \otimes U = U_1 \oplus U_2,
\end{align}
where the corresponding dimensions are given by $d_1 = d^2 -1 $ and $d_2=1$,
and the projectors onto the corresponding subspaces are $P_1 := I - \phi^+$ and $P_2 := \phi^+$.
From Schur's lemma, $Q$ can be decomposed as 
\begin{align}
Q = \sum_{k=1}^2 P_k^A \otimes Q_k^B,
\end{align}
and since $P_k^A$ are projectors, $Q$ is evaluated as
\begin{align}
Q &=\sum_{k=1}^2 \frac{P_k^A}{d_k} \otimes \mathrm{Tr}_{A} [ (P_k^A \otimes I^B) Q ] \\
&=\sum_{k=1}^2 \frac{P_k^A}{d_k} \otimes \mathrm{Tr}_{A} [ (P_k^A \otimes I^B) \pproj{Q'}^{AB} ],
\end{align}
where $\pproj{Q'}^{AB}$ is an arbitrary maximally entangled state between $A$ and $B$.
The second equality holds because of the partial trace on $A$.
Let the maximally entangled state $\kket{Q'}^{AB}$ be
\begin{align}
\kket{Q'}^{AB} = \sum_{l=1}^2 \sum_{\alpha=0}^{d_l-1} \ket{l,\alpha}^A \otimes \ket{l,\alpha}^B
\end{align}
where $l=1,2$ are the label for the irreducible representations and $\alpha$ for the basis in $P_l$.
Note that there is no multiplicity subspace in this case.
Then
\begin{gather}
(P_k^A \otimes I^B) \kket{Q'}^{AB} = \sum_{\alpha=0}^{d_k-1} \ket{k,\alpha}^{A} \otimes \ket{k,\alpha}^{B} , \\
\mathrm{Tr}_{A} [ (P_k^A \otimes I^B) \pproj{Q'}^{AB} ] = P_k^{B},
\end{gather}
and thus $Q$ can be written as 
\begin{align}
Q = \sum_{k=1}^2 \frac{1}{d_k} P_k^{A} \otimes P_k^{B} = \frac{1}{d^2-1} P_1^A \otimes P_1^B + P_2^A \otimes P_2^B.
\end{align}
The rank of $Q$ is $(d^2 - 1)^2 + 1$, and thus the dimension of $\mathrm{span} \{ J_U \}$ is $(d^2-1)^2 + 1$.

\end{proof}

By considering Lemma~\ref{lem:neutralization} and Lemma~\ref{lem:unitary2cptp},
it is enough to prove Theorem~\ref{thm:main2} in order to prove Theorem~\ref{thm:main}.

\begin{thm}\label{thm:main2}
Given a one-slot probabilistic comb $S_\stgs^{\ical_0 \ical_1 \ocal_1 \ocal_0}$ with $\dim \ical_1 = \dim \ocal_1 = d$
and $\dim \ical_0 = \dim \ocal_0 = d_0$.
If $S_\stgs^{\ical_0 \ical_1 \ocal_1} := \Tr_{\ocal_0} S_\stgs^{\ical_0 \ical_1 \ocal_1 \ocal_0}$ has a decomposition satisfying 
\begin{align}
S_\stgs^{\ical_0 \ical_1 \ocal_1}  =\frac{I^{\ical_0} }{ d_0 } \otimes \Tr_{\ical_0} S_\stgs^{\ical_0 \ical_1 \ocal_1} 
+ \sum_{i=1}^{d_0^2-1} \sum_{j=1}^{d^2-1} \alpha_{ij} h_i^{\ical_0} \otimes [ g_j^{\ical_1} \otimes I^{\ocal_1} ] 
+ \sum_{i=1}^{d_0^2-1} \sum_{j=1}^{d^2-1} \beta_{ij} h_i^{\ical_0} \otimes [ I^{\ical_1} \otimes g_j^{\ocal_1} ]
\end{align}
with coefficients $\{ \alpha_{ij} \}$ and  $\{ \beta_{ij} \}$,
then there exists $\varepsilon > 0$ and a $d$-slot comb $C = S + N$ satisfying
\begin{gather}
\Tr_{\ical \ocal} [S (J_U^{\otimes d})^T ] = \varepsilon \Tr_{\ical_1 \ocal_1} [ S_\stgs J_U^T ] \label{eq:condition_for_S_S1} \\
\Tr_{\ical \ocal} ( \psym N \psym ) \propto J_{id} \label{eq:condition_for_N}.
\end{gather}
\end{thm}

The proof of Theorem~\ref{thm:main2} contains two parts: 
the first part presents the construction of $\nio := \Tr_{\ocal_0} N$;
and the second part presents the construction of $N$ from $\nio$ by applying Lemma~\ref{lem:nio2n}.

\begin{proof}[Proof of Theorem~\ref{thm:main2}]

Let the Choi operator $S$ corresponds to success be
\begin{align}
S := \varepsilon S_\stgs^{\ical_0 \ical_{1} \ocal_{1} \ocal_0} \otimes \frac{ I^{\ical_2 \ocal_2} }{d} \otimes 
\cdots \otimes \frac{ I^{\ical_{d} \ocal_{d}} }{d}. \label{eq:ap_thm_proof_def_s}
\end{align}
Then the condition $S \geq0$ and Eq.~\eqref{eq:condition_for_S_S1} is satisfied.
The remaining conditions can be classified into the positivity condition $N \geq 0$, 
the causal condition that $C = S + N$ is a deterministic comb, and the neutralization condition Eq.~\eqref{eq:condition_for_N}.

We first show the idea to construct $N$ satisfying the causal condition.
One candidate of the Choi operator corresponding to failure, i.e., a Choi operator satisfies the causal condition that $C = S+F$ is a deterministic comb, is given by
\begin{align}
F := F^{\ical_0 \ical_{1} \ocal_{1}} \otimes \frac{ I^{\ical_2 \ocal_2} }{d} \otimes 
\cdots \otimes \frac{ I^{\ical_{d} \ocal_{d}} }{d} \otimes \frac{I^{\ocal_0}}{d}
\end{align}
where 
\begin{align}
F^{\ical_0 \ical_{1} \ocal_{1}} &:= \frac{I^{\ical_0 \ical_{1} \ocal_{1}}}{d} - \varepsilon S_\stgs^{\ical_0 \ical_{1} \ocal_{1}} \\
&= \frac{I^{\ical_0 \ical_{1} \ocal_{1}}}{d} - \varepsilon \biggl\{  \frac{I^{\ical_0} }{ d_0 } \otimes \Tr_{\ical_0} S_\stgs^{\ical_0 \ical_1 \ocal_1} 
+ \sum_{i=1}^{d_0^2-1} \sum_{j=1}^{d^2-1} \alpha_{ij} h_i^{\ical_0} \otimes [ g_j^{\ical_1} \otimes I^{\ocal_1} ]
+ \sum_{i=1}^{d_0^2-1} \sum_{j=1}^{d^2-1} \beta_{ij} h_i^{\ical_0} \otimes [ I^{\ical_1} \otimes g_j^{\ocal_1} ] \biggr\}.
\end{align}

This $F$ summed up with $S$ satisfies the causal condition by construction, 
but it does not satisfy the neutralization condition Eq.~\eqref{eq:condition_for_N}. 
Thus, it is enough to construct $N \geq 0$ that satisfies the following conditions
\begin{align}
\Tr_{\ocal_0} N - N^{(d)} \otimes \frac{I^{\ocal_{d}}}{d} &= 0 \label{eq:causal_condition_n_1} \\
\Tr_{\ical_k} N^{(k)} - N^{(k-1)} \otimes \frac{I^{\ocal_{k-1}}}{d} &= 0 \quad  (3 \leq k \leq d) \label{eq:causal_condition_n_2} \\
\Tr_{\ical_2} N^{(2)} - N^{(1)} \otimes \frac{I^{\ocal_1}}{d} &= d^{d-1} (F^{\ical_0 \ical_{1} \ocal_{1}}  - F^{\ical_0 \ical_{1}} \otimes \frac{I^{\ocal_1}}{d}) \label{eq:causal_condition_n_3} \\
\Tr_{\ical_1} N^{(1)} - (\Tr N ) \frac{I^{\ical_0}}{d_0}  &= 0 \label{eq:causal_condition_n_4} \\
\psym N \psym &= \frac{J_{id}^{\ical_0 \ocal_0}}{d_0} \otimes \Tr_{\ical_0 \ocal_0} [ \psym N \psym ], \label{eq:nsym_equals_jid}
\end{align}
where $N^{(d)} := \Tr_{\ocal_d \ocal_0} N$ and $N^{(k-1)} := \Tr_{\ocal_{k-1} \ical_k} N^{(k)}$ for $k = 2, \ldots, d$.

We divide the proof into two parts, by introducing the operator $\nio := \Tr_{\ocal_{0}} N$.
In the first part of the proof, we show the existence of $\nio$,
and the neutralization condition Eq.~\eqref{eq:nsym_equals_jid} is replaced by 
\begin{align}
\psym \nio \psym = \frac{I^{\ical_0}}{d_0} \otimes \Tr_{\ical_0} [ \psym \nio \psym ]. \label{eq:condition_for_nio}
\end{align}
In the second part of the proof (Lemma~\ref{lem:nio2n}), we construct the desired $N$ from $\nio$.
In both constructions, the following three conditions are considered: the positivity, the causal condition, and the neutralization condition.

(First part: construction of $\nio$)
Let $\nio$ be 
\begin{align}
\nio 
&:= \frac{1}{d} I^{\ical_0 \ical_{1} \ocal_{1}} \otimes \frac{ I^{\ical_2 \ocal_2} }{d} \otimes \cdots \otimes \frac{ I^{\ical_d \ocal_d} }{d} \notag \\
&- \varepsilon \biggl\{  \frac{I^{\ical_0} }{ d_0 } \otimes \Tr_{\ical_0} S_\stgs^{\ical_0 \ical_1 \ocal_1}  \otimes \frac{ I^{\ical_2 \ocal_2} }{d} \otimes \cdots \otimes \frac{ I^{\ical_d \ocal_d} }{d} \notag \\
&+ \sum_{i,j \geq 1} \alpha_{ij} h_i^{\ical_0} \otimes [ g_j^{\ical_1} \otimes I^{\ocal_1} ] \otimes \frac{ I^{\ical_2 \ocal_2} }{d} \otimes \cdots \otimes \frac{ I^{\ical_d \ocal_d} }{d} \notag \\
&+ \sum_{i,j \geq 1} (-\alpha_{ij}) h_i^{\ical_0} \otimes \frac{ I^{\ical_1 \ocal_1} }{d} \otimes [ g_j^{\ical_2} \otimes I^{\ocal_2} ] \otimes \cdots \otimes \frac{ I^{\ical_d \ocal_d} }{d} \notag \\
&+ \sum_{i,j \geq 1} \beta_{ij} h_i^{\ical_0} \otimes [ I^{\ical_1} \otimes g_j^{\ocal_1} ] \otimes \frac{ I^{\ical_2 \ocal_2} }{d} \otimes \cdots \otimes \frac{ I^{\ical_d \ocal_d} }{d} \notag \\
&+ \sum_{i,j \geq 1, \vec{k_2}} \beta_{ij} a_{2,\vec{k_2}} h_i^{\ical_0} \otimes [ g_{k_{2,1}}^{\ical_1} \otimes g_j^{\ocal_1} ] \otimes [ g_{k_{2,2}}^{\ical_2} \otimes \frac{I^{\ocal_2}}{d}] \otimes \frac{ I^{\ical_3 \ocal_3} }{d} \cdots \otimes \frac{ I^{\ical_d \ocal_d} }{d} \notag \\
& + \cdots + \notag \\
&+ \sum_{i,j \geq 1, \vec{k_d}} \beta_{ij} a_{d, \vec{k_d}} h_i^{\ical_0} \otimes [ g_{k_{d,1}}^{\ical_1} \otimes g_j^{\ocal_1} ] \otimes [ g_{k_{d,2}}^{\ical_2} \otimes \frac{I^{\ocal_2}}{d} ] \otimes \cdots \otimes [ g_{k_{d,d}}^{\ical_d} \otimes \frac{I^{\ocal_d}}{d}] \biggr\}, \label{eq:construction_nio}
\end{align}
where the summation on $\vec{k_m} = ( k_{m,1}, k_{m,2}, \ldots , k_{m,m} )$ denotes the summation on $ \{ k_{i,j} = 0,\ldots, d^2 -1\}$ for each term,
and coefficients $a_{m,\vec{k_m}}$ are determined in the following.

(Positivity)
The positivity of $\nio$ is trivial for small enough $\varepsilon$.
That is, since $\nio$ is of the form $\nio = I/ d^d + \varepsilon N'$ where $N'$ does not depend on $\varepsilon$,
there exists $\varepsilon > 0$ such that $\nio$ is strictly positive.

(Causal condition) 
Here we show that the causal conditions Eqs.~\eqref{eq:causal_condition_n_1}-\eqref{eq:causal_condition_n_4} are satisfied.
We first remark that the 1st, 2nd, 3rd and 5th lines sum up to $F$,
and we can write $\nio$ as $\nio = F + F'_s + \sum_{m=2}^d F'_m$ where $F'_s$ corresponds to the 4th line, 
and $F'_2, \ldots, F'_d$ correspond to the 6th to the last line.
Then it is enough to show that all $F' \in \{ F'_i \}_{i=s,2,3,\ldots,d}$ satisfies 
\begin{gather}
\Tr_{\ocal_0} F' -  F'^{(d)} \otimes \frac{I^{\ocal_{d}}}{d} = 0 \label{eq:causal_condition_fp_1}\\
\Tr_{\ical_k}  F'^{(k)} -  F'^{(k-1)} \otimes \frac{I^{\ocal_{k-1}}}{d} = 0 \quad  (2 \leq k \leq d) \label{eq:causal_condition_fp_2}\\
\Tr_{\ical_1}  F'^{(1)} - (\Tr  F' ) \frac{I^{\ical_0}}{d_0}  = 0 \label{eq:causal_condition_fp_3},
\end{gather}
where $F'^{(d)} := \Tr_{\ocal_d \ocal_0} F'$ and $F'^{(k-1)} := \Tr_{\ocal_{k-1} \ical_k} F'^{(k)}$ for $k = 2, \ldots, d$.

It is trivial that Eq.~\eqref{eq:causal_condition_fp_1} is satisfied for all $F'$.
It is also trivial to see that Eqs.~\eqref{eq:causal_condition_fp_2},\eqref{eq:causal_condition_fp_3} are satisfied for $F'_a$.
Thus, we consider Eqs.~\eqref{eq:causal_condition_fp_2},\eqref{eq:causal_condition_fp_3} for $F'_2, \ldots, F'_d$.
We can see that the l.h.s. of these equations are always of the form 
$\Tr_{\ical_k \cdots} ( F' - \Tr_{\ocal_{k-1}} F' \otimes \frac{I^{\ical_{k-1}}}{d}) $,
and $F'_m$ satisfies these conditions when $m < k$ because $F'_m$ has the term $\frac{I^{\ical_{k}}}{d}$ already.
In order to satisfy these conditions for $m \geq k$, we assume that the coefficients $a_{m,\vec{k_m}}$ satisfy 
\begin{align}
a_{m,\vec{k_m}} := a_{m,k_{m,1},k_{m,2},\ldots,k_{m,m}} = 0 \quad \mathrm{for} \quad k_{m,m} = 0, \label{eq:causal_requirement_for_a}
\end{align}
which is compatible with the following arguments on the neutralization condition.
By choosing these coefficients, $\Tr_{\ical_k \cdots} F'_m = 0$ is satisfied and Eqs.~\eqref{eq:causal_condition_fp_2},\eqref{eq:causal_condition_fp_3} is also satisfied.

(Neutralization condition)
Now we present a construction of coefficients $a_{m,\vec{k_m}}$ such that Eq.~\eqref{eq:condition_for_nio} is satisfied.
This condition is satisfied independently for the 1st line, 2nd line, the sum of 3rd and 4th lines, and the sum of the rest.
First, it is trivial that the 1st line and the 2nd line satisfies the condition, as it has $I^{\ical_0}$ in the system $\ical_0$. 
The sum of the 3rd and 4th lines vanishes on $\psym$, i.e., satisfies the condition with the r.h.s. being 0, 
because $\psym P_{\sigma} M P_{\sigma} \psym  = \psym M \psym$ holds for any permutation $\sigma$ and an arbitrary operator $M$.
For the sum of 5th line and after, we see that for each $i,j \geq 1$, it can be written as $\beta_{ij} h_i^{\ical_0} \otimes C_j$ with
\begin{align}
C_j &:=  [ I^{\ical_1} \otimes g_j^{\ocal_1} ] \otimes \frac{ I^{\ical_2 \ocal_2} }{d} \otimes \cdots \otimes \frac{ I^{\ical_d \ocal_d} }{d} \notag \\
&+ \sum_{\vec{k_2}}  a_{2,\vec{k_2}} [ g_{k_{2,1}}^{\ical_1} \otimes g_j^{\ocal_1} ] \otimes [ g_{k_{2,2}}^{\ical_2} \otimes \frac{I^{\ocal_2}}{d}] \otimes \frac{ I^{\ical_3 \ocal_3} }{d} \cdots \otimes \frac{ I^{\ical_d \ocal_d} }{d} \notag \\
& + \cdots + \notag \\
&+ \sum_{\vec{k_d}} a_{d, \vec{k_d}} [ g_{k_{d,1}}^{\ical_1} \otimes g_j^{\ocal_1} ] \otimes [ g_{k_{d,2}}^{\ical_2} \otimes \frac{I^{\ocal_2}}{d} ] \otimes \cdots \otimes [ g_{k_{d,d}}^{\ical_d} \otimes \frac{I^{\ocal_d}}{d}] \\
&= [ I^{\ical_1} \otimes I^{\ical_2} \otimes \cdots \otimes I^{\ical_d} 
+ \sum_{\vec{k_2}} a_{2, \vec{k_2}}  g_{k_{2,1}}^{\ical_1} \otimes g_{k_{2,2}}^{\ical_2} \otimes I^{\ical_3} \otimes \cdots \otimes I^{\ical_d} \notag \\
& \quad + \cdots 
+ \sum_{\vec{k_d}} a_{d, \vec{k_d}} g_{k_{d,1}}^{\ical_1} \otimes g_{k_{d,2}}^{\ical_2} \otimes g_{k_{d,3}}^{\ical_3} \otimes \cdots \otimes g_{k_{d,d}}^{\ical_d} ] \notag \\
& \otimes [ g_j^{\ocal_1} \otimes \frac{I^{\ocal_2}}{d} \otimes \cdots \otimes \frac{I^{\ocal_d}}{d} ]. \label{eq:cj_iii_otimes_gii}
\end{align}
In the following, we show that the neutralization condition is satisfied for each $i,j$, by showing that $C_j$ vanishes on $\psym$ as $\psym C_j \psym = 0$.

Here, we choose the coefficients $\{a_{m,\vec{k_m}} \}$ such that the first term is the $d$ qudit (unnormalized) totally antisymmetric state $d^d A_d = d^d \proj{A_d}$. 
These coefficients are available as follows.
Note that we assumed Eq.~\eqref{eq:causal_requirement_for_a} in the causal condition part.
Since $g_{i_1} \otimes g_{i_2} \otimes \cdots g_{i_d}$ forms a basis, any operator including $A_d$ can be written as
$\sum_{i_1, i_2, \ldots, i_d} a_{i_1, i_2, \ldots, i_d} g_{i_1} \otimes g_{i_2} \otimes \cdots \otimes g_{i_d}$.
However, the coefficients $\{a_{m,\vec{k_m}} \}$ has the constraint Eq.~\eqref{eq:causal_requirement_for_a} and cannot cover arbitrary operator.
Especially, it lacks the terms $g_{i_1} \otimes I \otimes \cdots \otimes I$ with $i_1 \neq 0$.
The totally antisymmetric state satisfies $\Tr_{2,\ldots,d} A_d = I_1$, 
and the coefficients corresponding to these terms that containing only one traceless operator $g_{i_1}$ are actually 0.
Thus, there exists a set of $\{a_{m,\vec{k_m}} \}$ satisfying Eq.~\eqref{eq:causal_requirement_for_a} 
and that Eq.~\eqref{eq:cj_iii_otimes_gii} can be evaluated as 
\begin{align}
C_j = d^d A_d^{\ical} \otimes [ g_j^{\ocal_1} \otimes \frac{I^{\ocal_2}}{d} \otimes \cdots \otimes \frac{I^{\ocal_d}}{d} ]
=: d A_d^\ical \otimes M_j^{\ocal}.
\end{align}
Now we show that $C_j$ vanishes on $\psym$.
Consider that $\psym = \sum_\sigma P^{\ical \ocal}_\sigma = \sum_\sigma P^{\ical}_\sigma \otimes P^{\ocal}_\sigma$
and $P^{\ical}_\sigma \ket{A_d} = sgn(\sigma) \ket{A_d}$, $\psym (A_d^\ical \otimes M_j^\ocal) \psym$ can be evaluated as 
\begin{align}
\psym (A_d^\ical \otimes M_j^\ocal) \psym 
&= A_d^\ical \otimes  [ \sum_{\sigma} sgn(\sigma)  P^{\ocal}_\sigma ] M_j^\ocal [ \sum_{\sigma'} sgn(\sigma') P^{\ocal}_{\sigma'} ] \\
&= A_d^\ical \otimes A_d^\ocal M_j^\ocal A_d^\ocal 
\end{align}
Also,
\begin{align}
\Tr\, A_d^\ocal M_j^\ocal A_d^\ocal
&= \bra{A_d}  g_j^{\ocal_1} \otimes I^{\ocal_2} \otimes \cdots \otimes I^{\ocal_d} \ket{A_d} \\
&=\Tr\, g_j^{\ocal_1} = 0
\end{align}
because $g_j^{\ocal_1}$ are traceless for $j\geq 1$. 
Thus, we obtain $A_d^\ocal M_j^\ocal A_d^\ocal = 0$ and $\psym C_j \psym = 0$ for $j \geq 1$.

(Second part: construction of $N$ from $\nio$)
We apply Lemma~\ref{lem:nio2n}.
The operator $d^d \nio = I + \varepsilon N'$ corresponds to $M^{AB} = I + \varepsilon M'$, $d^{d+1} N$ corresponds to $M^{ABC}$,
and systems $A,B,C$ correspond to $\ical_0, \ical\otimes\ocal, \ocal_0$ respectively.

\end{proof}

\begin{lem}\label{lem:nio2n}
Let $\hcal_A, \hcal_B, \hcal_C \simeq \hcal_A$ be Hilbert spaces with dimensions $d_0, d_B, d_0$, 
$\Pi^B$ be a projector on $L(\hcal_B)$, and $J_{id}^{AC} = d_0 \phi^+$ be the maximally entangled state on $\hcal_A \otimes \hcal_C$.
Given an operator $M' \in L(\hcal_A \otimes \hcal_B)$,
there exists $\varepsilon > 0$ such that the following holds.
If $M^{AB} = I + \varepsilon M'$ satisfies
\begin{gather}
M^{AB} \geq 0 \\
\Pi^B M^{AB} \Pi^B = \frac{I^A}{d_0} \otimes \Tr_{A} \Pi^B M^{AB} \Pi^B, \label{eq:lem_cond_id}
\end{gather}
there exists an operator $M^{ABC} \in L(\hcal_A \otimes \hcal_B \otimes \hcal_C)$ satisfies
\begin{gather}
M^{ABC} \geq 0 \\
\Tr_C M^{ABC} = M^{AB} \label{eq:lem_def_nio} \\
\Pi^B M^{ABC} \Pi^B = \frac{1}{d_0} J_{id}^{AC} \otimes \Tr_{AC} \Pi^B M^{ABC} \Pi^B. \label{eq:lem_neutralization_cond}
\end{gather}
\end{lem}

\begin{proof}
Let $\{ h_i \}$ with $h_0 = I$ be a Hermitian basis for $\hcal_A$ and $\hcal_C$.
Let $M_i^B := \frac{1}{d_0} \Tr_{A} h_i^A M^{AB}$, so that $M^{AB} = \sum_i h_i^A \otimes M^B_{i}$ holds.
Note that with this decomposition, the condition Eq.~\eqref{eq:lem_cond_id} is given by $\Pi^B M^{AB} \Pi^B = {I^A} \otimes \Pi^B M_0^B \Pi^B$ and $\Pi^B M_i^B \Pi^B = 0$ for $i \neq 0$.

For simplicity of the proof, we give a construction of $M^{ABC}$ first as
\begin{align}
M^{ABC} &:=  J_{id}^{AC} \otimes \Pi^B M_0^B \Pi^B
+  \frac{1}{d_0} (I^A \otimes I^C) \otimes \Pi_\perp^{B} M_0^B \Pi_\perp^{B} \notag \\
&+ \frac{1}{d_0} \sum_{i \geq 1} h_i^A \otimes \Pi_\perp^{B} M_i^B \Pi_\perp^{B} \otimes I^C \notag \\
&+ \frac{1}{d_0} \sum_{k \geq 0} ( h_k^A \otimes I^C + \sum_{i \geq 0, j \geq 1} \alpha_{ijk} h_i^A \otimes h_j^C ) \otimes \Pi^{B} M_k^B \Pi_\perp^{B} \notag \\
&+ \frac{1}{d_0} \sum_{k \geq 0} ( h_k^A \otimes I^C + \sum_{i \geq 0, j \geq 1} \alpha^*_{ijk} h_i^A \otimes h_j^C ) \otimes \Pi_\perp^{B} M_k^B \Pi^{B}, \label{eq:construction_mabc}
\end{align}
where $\{ \alpha_{ijk} \}$ are complex numbers determined in the following.
It is easy to see that the causal condition Eq.~\eqref{eq:lem_def_nio} and the neutralization condition Eq.~\eqref{eq:lem_neutralization_cond} are satisfied,
and the remaining condition for $M^{ABC}$ is the positivity.

In order to guarantee the positivity, we first consider the support given by the projector
\begin{align}
\psup = (\phi^+)^{AC} \otimes \Pi^B + I^{AC} \otimes \Pi_\perp^{B}
\end{align}
with the projector $\phi^+ = J_{id} / d_0$, then obtain parameters $\{ \alpha_{ijk} \}$ so that $M^{ABC}$ is on this support, 
and finally show that $M^{ABC}$ is positive with small enough $\varepsilon$.
The condition $\psup M^{ABC} \psup = M^{ABC}$ is satisfied if the following holds
\begin{align}
(\phi^+)^{AC} ( h_k^A \otimes I^C + \sum_{i \geq 0, j \geq 1} \alpha_{ijk} h_i^A \otimes h_j^C ) I^{AC}
= ( h_k^A \otimes I^C + \sum_{i \geq 0, j \geq 1} \alpha_{ijk} h_i^A \otimes h_j^C ),
\end{align}
or equivalently 
\begin{align}
\phi^+ A_k = A_k
\end{align}
with $A_k := h_k^A \otimes I^C + \sum_{i \geq 0, j \geq 1} \alpha_{ijk} h_i^A \otimes h_j^C $.
Since $\{ \alpha_{ijk} \}$ can be any complex numbers, the restrictions for $\{ A_k \}$ are given by
\begin{align}
\Tr (h_{k'} \otimes I) A_{k} = d_0^2 \delta_{k k'}
\end{align}
for all $k,k'$.
In order to satisfy $\phi^+ A_k = A_k$, $A_{k}$ should be decomposed as $ A_{k} = \ketbra{\phi^+}{a_{k}}$,
where $\ket{a_k}$ is an unnormalized vector.
Let $\ket{a_{k}} = \sum_{m,n = 0}^{d_0 - 1} a^{(k)}_{mn} \ket{mn}$,
then the condition for $a^{(k)}_{mn}$ is that 
\begin{gather}
\Tr (h_{k'} \otimes I) A_{k} = \sum_{m,n = 0}^{d_0 - 1} (a^{(k)}_{mn})^* \bra{m} h_{k'} \ket{n} = d_0^2 \delta_{k k'},
\end{gather}
for $k,k' = 0, \ldots, d_0^2 - 1 $. %
Here, the $d_0^2$ parameters $a_{mn}$ can be chosen freely, and there are $d_0^2$ linear (and independent due to the orthogonality of $h_{k'}$) constraints, 
thus, there exists a feasible $a_{mn}$, $A_k$, and $\alpha_{ijk}$.
Thus, $\psup M^{ABC} \psup = M^{ABC}$ holds.

For $M^{AB} = I$, a possible $M^{ABC}$ is given by 
\begin{align}
M^{ABC} &= J_{id}^{AC} \otimes \Pi^B + \frac{1}{d_0} I^{AC} \otimes \Pi_\perp^B =: M_{0}^{ABC}
\end{align}
For $M^{AB} = I + \varepsilon M'$, the corresponding $M^{ABC}$ can be written as
\begin{align}
M^{ABC} &= M_0^{ABC} + \varepsilon M''
\end{align}
where $M''$ is an operator only depends on $M'$, because the construction of $M^{ABC}$ given by Eq.~\eqref{eq:construction_mabc} is linear in $M^{AB}$.
The non-zero minimum eigenvalue is given by
\begin{align}
\min_{\ket{\psi} \in \psup} \bra{\psi} M^{ABC} \ket{\psi} = \min_{\ket{\psi} \in \psup } [ \bra{\psi} M_0^{ABC} \ket{\psi} + \varepsilon \bra{\psi}  M'' \ket{\psi} ],
\end{align}
because $\psup M^{ABC} \psup = M^{ABC}$ is satisfied.
The minimum eigenvalue on the support $\psup$ is given by minimizing the $\ket{\psi}$ with vectors only on $\psup$,
in which case the first term is strictly positive, especially larger than $1/d_0$.
Thus, there exists $\varepsilon>0$ such that the minimum eigenvalue on $\psup$ is greater than 0, and the positivity of $M^{ABC}$ is guaranteed.

\end{proof}

\begin{rem}\label{rem:epsilon}
The construction of the Choi operators presented in the proof of Theorem~\ref{thm:main2} does not provide the optimal success probability in general.
The success probability for success-or-draw depends on $\varepsilon$, and the construction presented only focuses on the existence of a non-zero $\varepsilon$.
While it is difficult to obtain a general lower bound on $\varepsilon$, here we provide some insight into its value.

The first observation is that the construction of $S$ in Eq.~\eqref{eq:ap_thm_proof_def_s} indicates that this supermap only utilizes a single copy of the input unitary operation although $d$ copies are given,
and thus this construction must provide $\varepsilon \leq 1$.
In general, $\varepsilon$ can be greater than $1$.
For example, the numerical results on unitary inversion presented in the main text show that $\varepsilon = 4/3 > 1$ is achievable.

The second observation is that the positivity condition of the Choi operators strongly limits the value of $\varepsilon$.
In order to guarantee the positivity of $\nio$ in Eq.~\eqref{eq:construction_nio}, 
we can require $\varepsilon$ to be small enough that $\nio$ is a diagonally dominant matrix.
However, a lower bound on $\varepsilon$ obtained via this method is usually much smaller than the optimal value.
In general, there are few methods for determining the positivity of an operator, especially if the dimension is large.

\end{rem}

\begin{rem}\label{rem:indefinite_causal}
In the second part for the proof of Theorem~\ref{thm:main2} (mostly equivalent to Lemma~\ref{lem:nio2n}),
the condition Eq.~\eqref{eq:causal_condition_n_1} (Eq.~\eqref{eq:lem_def_nio}) is assumed which corresponds to the causal condition that the corresponding Choi operator is a sequential supermap or quantum comb.
However, when the indefinite causal order~\cite{indefinite1,indefinite2,indefinite3,indefinite_purification} is allowed,
this causal condition can be relaxed and the construction of $N$ from $\nio$ can be replaced as follows instead of Lemma~\ref{lem:nio2n}.
The conditions for an indefinite causal order supermap are that the corresponding Choi operator is positive, and that when the input operations are CPTP maps, the output operation is also a CPTP map.
Here we consider a subset of such supermaps which Choi operators satisfy the following condition:
\begin{align}
\Tr_{\ocal_0} C = \sum_\sigma p_\sigma C_\sigma^{\ical_0 \ical \ocal} \label{eq:ap_rem_def_c}
\end{align}
where $p_\sigma$ are probabilities sum up to 1, 
and $C_\sigma^{\ical_0 \ical \ocal}$ denotes a sequential supermap where the order of input operations are permuted with respect to the permutation $\sigma$.
This is a strictly stronger condition than that of the indefinite causal order supermaps,
but many quantum supermaps satisfy this condition such as the quantum switch.

Let $\nio_{\sigma} := P^{\ical \ocal}_\sigma (\nio) P^{\ical \ocal}_\sigma$ be the probabilistic comb with the order of input operations permuted by $\sigma$.
We define $N$ as
\begin{align}
N &:= ( \frac{1}{N!} \sum_\sigma \nio_\sigma ) \otimes \frac{ I^{\ocal_0} }{d_0}
+ \frac{1}{d_0}  \sum_{ij \geq 1} \eta_{ij} h_i^{\ical_0} \psym ( \frac{1}{N!} \sum_\sigma \nio_\sigma ) \psym \otimes h_j^{\ocal_0} \label{eq:ap_rem_n1} \\
&= \frac{I^{\ical_0}}{d_0} \otimes \frac{1}{N!} \sum_\sigma (\Tr_{\ical_0} \psym \nio_\sigma \psym) \otimes \frac{ I^{\ocal_0} }{d_0}
+ \frac{1}{N!} \sum_\sigma \pnsym  ( \nio_\sigma ) \pnsym \otimes \frac{ I^{\ocal_0}}{d_0} \notag \\
&\quad+ \sum_{ij \geq 1} \eta_{ij} \frac{h_i^{\ical_0}}{d_0} \otimes \frac{1}{N!}\sum_\sigma(\Tr_{\ical_0} \psym \nio_\sigma \psym) \otimes \frac{ h_j^{\ocal_0} }{d_0} \\
&= \frac{1}{d_0} J_{id}^{\ical_0 \ocal_0} \otimes \frac{1}{N!} \sum_\sigma (\Tr_{\ical_0} \psym \nio_\sigma \psym)
+  \frac{1}{N!} \sum_\sigma \pnsym ( \nio_\sigma ) \pnsym \otimes \frac{ I^{\ocal_0} }{d_0}
\end{align}
where the coefficients $\eta_{ij}$ are determined by $J_{id} = \frac{1}{d_0} I \otimes I + \frac{1}{d_0}\sum_{ij \geq 1} \eta_{ij} h_i \otimes h_j$.
In the first equality, we also use the fact that if an operator is permutation invariant, it is block diagonal in $\psym$ and $\pnsym$,
that is, the off-diagonal terms vanish as
\begin{align}
\psym (\frac{1}{N!} \sum_\sigma \nio_\sigma) \pnsym 
&= \psym (\frac{1}{N!} \sum_\sigma P^{\ical \ocal}_\sigma (\nio) P^{\ical \ocal}_\sigma) (I - \psym)\\
&= \psym \frac{1}{N!} \sum_\sigma (\nio) (P^{\ical \ocal}_\sigma - \psym)\\
&= \psym (\nio) (\psym - \psym) = 0.
\end{align}
By this construction, the positivity of $N$ is preserved because both terms in Eq.~\eqref{eq:ap_rem_n1} are positive, 
and the neutralization condition $\psym N \psym = J_{id} / d_0 \otimes \Tr_{\ical_0 \ocal_0} \psym N \psym$ is also satisfied.
To see the causal condition can be satisfied, we first note that 
\begin{align}
\Tr_{\ocal_0} N = \frac{1}{N!} \sum_\sigma \nio_\sigma
\end{align}
holds. 
Since there exists an operator $S$ such that $\Tr_{\ocal_0} (S + N)$ satisfies the sequential condition (which is actually given by Eq.~\eqref{eq:ap_thm_proof_def_s}),
by defining $S_\sigma^{\ical_0 \ical \ocal} := P^{\ical \ocal}_\sigma (\Tr_{\ocal_0} S) P^{\ical \ocal}_\sigma$,
$C_\sigma^{\ical_0 \ical \ocal} := S_\sigma^{\ical_0 \ical \ocal} + \nio_\sigma$ and $p_\sigma = 1/N!$,
the causal condition Eq.~\eqref{eq:ap_rem_def_c} is satisfied.

\end{rem}

\section{Success-or-draw is not possible with a single call for unitary inversion}\label{ap:proof2}

For the two-dimensional unitary inversion, 
we show that it is not possible to have a success-or-draw protocol if we have only a single use of the input unitary operation.
Especially, we show the only solution to the following SDP is $p=0$. Note that we denote $d=2$ in order to clarify that it corresponds to the dimension.
\begin{gather}
\max\quad p \\
\mathrm{s.t.}\quad \Tr_{\ical_1 \ocal_1} [ S {J_{U}}^T ] = p {J_{U^{-1}}} \label{eq:ap_success_inverse}\\
\Tr_{\ical_1 \ocal_1} [ N {J_{U}}^T ] \leq d {J_{id}} \label{eq:ap_neutralization_inverse}\\
S \geq 0, N \geq 0\\
\Tr_{\ocal_0} (S+N) = \Tr_{\ocal_1 \ocal_0} (S+N) \otimes \frac{ I^{\ocal_1} }{ d } \\
\Tr_{\ical_1 \ocal_1 \ocal_0} (S+N) = \Tr (S+N)  \frac{ I ^{\ical_0} }{ d }
\end{gather}

\begin{proof}

Assuming that $\{ p, S, N \}$ is a solution to this SDP,
then for any $U$, $\{ p, (U^{\ical_1} \otimes U^{\ocal_0}) S (U^{\ical_1} \otimes U^{\ocal_0}), U^{\ical_1} N U^{\ical_1} \}$ is also a solution to this SDP, because it satisfies all of the conditions.
By defining $S' = \int dU (U^{\ical_1} \otimes U^{\ocal_0}) S (U^{\ical_1} \otimes U^{\ocal_0})$ and $ N' = \int dU U^{\ical_1} N U^{\ical_1} $,
we obtain $\{ p, S', N' \}$ which is also a solution to this SDP.
Thus, without loss of generality, we can assume the following commutation relation 
\begin{gather}
[ S, U^{\ical_1} \otimes U^{\ocal_0} ] = 0 \label{eq:ap_comm_1}\\
[ N, U^{\ical_1} ] = 0. \label{eq:ap_comm_2}
\end{gather}

From the second commutation relation Eq.\eqref{eq:ap_comm_2} and Schur's lemma, $N$ can be decomposed as 
\begin{align}
N = N^{\ical_0 \ocal_1 \ocal_0} \otimes \frac{ I^{\ical_1} }{d}.
\end{align}
Consider Eq.~\eqref{eq:ap_neutralization_inverse} with $U = I$, we obtain
\begin{align}
d J_{id} &\geq \Tr_{\ical_1 \ocal_1} [ (N^{\ical_0 \ocal_1 \ocal_0} \otimes \frac{ I^{\ical_1} }{d}) {J_{id}}^T ] \\
& = \Tr_{\ocal_1} [ N^{\ical_0 \ocal_1 \ocal_0} ], \label{eq:ap_jid_tr_n013}
\end{align}
and $N^{\ical_0 \ocal_1 \ocal_0} $ can be decomposed as
\begin{align}
N^{\ical_0 \ocal_1 \ocal_0} = N^{\ocal_1} \otimes J_{id}^{\ical_0 \ocal_0} / d \label{eq:ap_n_product_state}
\end{align}
as follows.
Let $N^{\ical_0 \ocal_1 \ocal_0} = \sum_i p_i \proj{n_i^{\ical_0 \ocal_1 \ocal_0}}$.
Since $J_{id}$ is rank-1, Eq.~\eqref{eq:ap_jid_tr_n013} indicates that
$\Tr_{\ocal_1} \proj{n_i^{\ical_0 \ocal_1 \ocal_0}} \propto J_{id}$ holds for all $i$.
Consider the Schmidt decomposition $\ket{n_i^{\ical_0 \ocal_1 \ocal_0}} = \sum_j \alpha_{ij} \ket{a_j^{\ical_0 \ocal_0}} \otimes \ket{b_j^{\ocal_1}} $,
then $\Tr_{\ocal_1} \proj{n_i^{\ical_0 \ocal_1 \ocal_0}} = \sum_j |\alpha_{ij}|^2 \proj{a_j^{\ical_0 \ocal_0}}$ is proportional to the rank-1 operator$J_{id}$,
which means the only possible solution is that $\ket{n_i^{\ical_0 \ocal_1 \ocal_0}} = \ket{(\phi^+)^{\ical_0 \ocal_0}} \otimes \ket{b_j^{\ocal_1}}$ where $\proj{\phi^+} = J_{id} / d$ is the maximally entangled state.
Thus, $N^{\ical_0 \ocal_1 \ocal_0}$ can be decomposed as Eq.~\eqref{eq:ap_n_product_state}.

On the other hand, we can show
\begin{align}
S = p J_Y^{\ical_0 \ocal_1} \otimes J_Y^{\ical_1 \ocal_0 }\label{eq:ap_s_eq_jyjy}
\end{align}
as follows.
Note that $J_Y = d \psi^- = d \proj{\psi^-}$ where 
$\ket{\psi^-} = ( 1/\sqrt{2} ) ( \ket{01} - \ket{10} )$ is a maximally entangled state also known as the singlet state.
From Eq.\eqref{eq:ap_comm_1} and Schur's lemma, $S$ can be decomposed as
$S = S^{\ical_0 \ocal_1} \otimes { J_{Y}^{\ical_1 \ocal_0} } / {d}$.
Let $S^{\ical_0 \ocal_1} = \sum_i p_i \proj{s_i^{\ical_0 \ocal_1}}$ and consider Eq.~\eqref{eq:ap_success_inverse}.
Since the r.h.s. of Eq.~\eqref{eq:ap_success_inverse} is rank-1, it is necessary for every $i$ that
\begin{align}
\Tr_{\ical_1 \ocal_1} [ (\proj{s_i^{\ical_0 \ocal_1}} \otimes \frac{ J_{Y}^{\ical_1 \ocal_0} }{d} )  {J_{id}} ] \propto J_{id} \label{eq:ap_pure_success}
\end{align}
holds, where we choose $U = I$ in Eq.~\eqref{eq:ap_success_inverse}.
Consider the Schmidt decomposition $\ket{s_i^{\ical_0 \ocal_1}} = \sum_j \alpha_{ij} \ket{a_j}^{\ical_0} \otimes Y \ket{b_j}^{\ocal_1}$, 
where $\{ \ket{a_j} \}$ and $\{ \ket{b_j} \}$ are some basis and the Pauli operator $Y$ is added for convenience.
Then Eq.~\eqref{eq:ap_pure_success} becomes
\begin{align}
\sum_j \alpha_{ij} \ket{a_j}^{\ical_0} \otimes \ket{b_j}^{\ocal_0} \propto \ket{\phi^+}^{\ical_0 \ocal_0}.
\end{align}
and thus $\proj{s_i^{\ical_0 \ocal_1}}$ is proportional to $J_Y$.
The constant factor is obtained by direct calculation, and Eq.~\eqref{eq:ap_s_eq_jyjy} is proved.

By using the causal conditions, we obtain
\begin{align}
\Tr_{\ical_1 \ocal_0} (S+N) = \Tr_{ \ical_1  \ocal_1 \ocal_0} (S+N) \otimes \frac{ I^{\ocal_1} }{ d } 
= \Tr (S+N) \frac{ I^{\ical_0} \otimes I^{\ocal_1} }{ d^2 } = I^{\ical_0} \otimes I^{\ocal_1}
\end{align}
and since $S$ is given by Eq.~\eqref{eq:ap_s_eq_jyjy}, we obtain
\begin{align}
N^{\ical_0 \ocal_1} = I^{\ical_0 \ocal_1} - p d J_Y^{\ical_0 \ocal_1}. \label{eq:ap_n_eq_id_minus_jy}
\end{align}
On the other hand, Eq.~\eqref{eq:ap_n_product_state} indicates $N^{\ical_0 \ocal_1 } = N^{\ocal_1} \otimes I^{\ical_0} / d$,
and the only possible solution with Eq.~\eqref{eq:ap_n_eq_id_minus_jy} is $p=0$.

\end{proof}

\end{document}